\documentclass[journal, onecolumn, draftclsnofoot, twoside, 11 pt]{IEEEtran}
\IEEEoverridecommandlockouts
\usepackage{cite}
\usepackage{amsmath,amssymb,amsfonts}
\usepackage{amsthm}
\usepackage{graphicx}
\usepackage{textcomp}
\usepackage{xcolor}
\usepackage[linesnumbered,ruled,vlined]{algorithm2e}
\usepackage{etoolbox}

\newtheorem{thrm}{Theorem}
\newtheorem{lemma}{Lemma}

\newtheorem{ex}{Example}
\newtheorem{prop}{Proposition}
\newtheorem{defi}{Definition}
\newtheorem{coro}{Corollary}
\DeclareMathOperator*{\argmin}{arg\,min}
\DeclareMathOperator*{\argmax}{arg\,max}
\def\BibTeX{{\rm B\kern-.05em{\sc i\kern-.025em b}\kern-.08em
    T\kern-.1667em\lower.7ex\hbox{E}\kern-.125emX}}

\newcommand{\overbar}[1]{\mkern 1.5mu\overline{\mkern-1.5mu#1\mkern-1.5mu}\mkern 1.5mu}    
\begin{document}

\title{Minimum Age TDMA Scheduling\thanks{This paper was presented in part 
at IEEE INFOCOM 2019~\cite{conf}.}\thanks{This work is partially supported 
by the Ministry of Science and Technology of R.O.C. under contract No. 
MOST 106-2221-E-004-005-MY3.}}

\author{\IEEEauthorblockN{Tung-Wei Kuo}\\
\IEEEauthorblockA{\textit{Department of Computer Science} \\
\textit{National Chengchi University}\\
Taipei, Taiwan \\
twkuo@cs.nccu.edu.tw}
}

\maketitle
\begin{abstract}
We consider a transmission scheduling problem in which
multiple systems receive update information through a shared 
Time Division Multiple Access (TDMA) channel. 
To provide timely delivery of update information,
the problem asks for a schedule that minimizes the overall age of information.
We call this problem the Min-Age problem. 
This problem is first studied by He \textit{et al.} 
[IEEE Trans. Inform. Theory, 2018], 
who identified several special cases where the problem can be solved 
optimally in polynomial time. Our contribution is threefold. 
First, we introduce a new job scheduling problem called the Min-WCS problem, 
and we prove that, for any constant $r \geq 1$, 
every $r$-approximation algorithm for the Min-WCS problem can be transformed 
into an $r$-approximation algorithm for the Min-Age problem. 
Second, we give a randomized 2.733-approximation algorithm 
and a dynamic-programming-based exact algorithm for the Min-WCS problem. 
Finally, we prove that the Min-Age problem is NP-hard. 
\end{abstract}


\section{Introduction}
We consider systems whose states change upon reception of update messages. 
Such systems include, for example, web caches~\cite{Yu:1999:SWC:316188.316219}, 
intelligent vehicles~\cite{5307471}, and real-time databases~\cite{1275297}. 
The timely delivery of update messages is often critical to the smooth and
secure functioning of the system. Moreover, since any given update is 
likely dependent on previous updates, the update messages should not be 
delivered out of order. In most cases, the system does not have 
exclusive access to a communication channel. 
Instead, it must share the channel with other systems. 
Hence, the transmission schedule 
plays a crucial role in determining the performance 
of the systems that share the channel.

This scenario can be modeled by multiple sender-receiver pairs 
and a channel shared by these sender-receiver pairs. The sender 
sends update messages to the receiver through the shared channel, 
and the receiver changes its state upon reception of an update message.\footnote{The 
sender may serve as a relay or hub for the system and thus may not be 
responsible for generating update messages.} 
This paper discusses the design of transmission scheduling algorithms 
for such channels. 
Specifically, we assume that the channel has a buffer in which the update 
messages are stored, and a transmission schedule for the messages 
in the buffer must be determined.\footnote{The buffer 
may be a logical one that stores the inputs to a scheduler.} 
In this paper, we refer to a system that changes its state upon reception 
of an update message as a \textit{receiver}.

To keep the state of a receiver as fresh as possible, 
it is important to keep the age of the receiver as small as possible.
Specifically, the age of a receiver is the age of the receiver's  
most recently received message $M$, i.e.,
the difference between the current time and the time at which $M$ is generated.
Most prior research analyzes 
the age of a receiver through stochastic 
process models~\cite{5984917, 6195689, 7282742, 6620189, bestpaper, 8006590, 
7852321, DBLP:journals/corr/KaulY17, DBLP:journals/corr/abs-1803-06469},
where the randomness comes from the state of the channel or 
the arrival process of update messages. 
In this paper, we take a combinatorial optimization approach to 
minimize the overall age of all receivers on a reliable channel.
In particular, we study the problem defined by He \textit{et al.},
who considered a scenario in which the transmission 
scheduling algorithm is invoked repeatedly~\cite{IT}. 
Specifically, after the scheduling algorithm computes a schedule, 
the channel then delivers the messages according to the schedule.
New messages may arrive while the channel is delivering the scheduled messages. 
These new messages are stored in the buffer and scheduled for transmission 
during the next invocation of the algorithm.

The scheduling algorithm should be designed with the characteristics of the 
channel in mind. For example, He \textit{et al.} considered a wireless channel, 
in which various senders might interfere with one another~\cite{IT}.
They also considered a Time Division Multiple Access (TDMA) channel, 
in which the channel delivers one message at a time.
They identified some conditions in which optimal schedules can 
be obtained by sorting the sender-receiver pairs 
according to the number of messages to be sent to the receiver~\cite{IT}.
However, even if the channel is TDMA-based, 
it remained open whether the problem can be solved 
optimally in polynomial time. In this paper, we therefore 
focus on TDMA channels. In the remainder of this paper, 
we refer to this scheduling problem on 
a TDMA channel as the Min-Age problem.

In this paper, we cast the Min-Age problem as a job scheduling problem called 
the Min-WCS problem. 
The Min-WCS problem has a simple formulation 
inspired by a geometric interpretation of the Min-Age problem.
The simplicity of the formulation also 
facilitates algorithm design. 
As we will see in Section~\ref{sec: conclusion}, 
one may solve variants of the Min-Age problem by modifying the 
geometric interpretation and then solving the corresponding job scheduling 
problem. 

Job scheduling has been studied for decades. 
In fact, the Min-WCS problem is a special case 
of single-machine scheduling 
with a non-linear objective function under precedence constraints, which has been 
studied by Schulz and Verschae~\cite{schulz_et_al:LIPIcs:2016:6415} 
and Carrasco \textit{et al.}~\cite{CARRASCO2013436}. 
Specifically, for any $\epsilon>0$, the algorithm proposed by Schulz and Verschae 
approximates the optimum within a factor of $(2+\epsilon)$ 
when the objective function is concave~\cite{schulz_et_al:LIPIcs:2016:6415}. 
When the objective function is convex, Carrasco \textit{et al.} proposed
a $(4+\epsilon)$-speed 1-approximation algorithm for any 
$\epsilon > 0$~\cite{CARRASCO2013436}.\footnote{Specifically, 
let $OPT$ be the optimal objective value.
An $s$-speed $r$-approximation algorithm for a minimization problem
finds a solution of objective value at most $r \cdot OPT$ when using a machine 
that is $s$ times faster than the original machine.}
The solutions proposed by 
Schulz and Verschae~\cite{schulz_et_al:LIPIcs:2016:6415} 
and Carrasco \textit{et al.}~\cite{CARRASCO2013436} are based on linear 
programming rounding. The objective function of the Min-WCS 
problem is convex, and we give a randomized 
2.733-approximation algorithm for the Min-WCS problem without 
linear programming.
We summarize our major results as follows:

\noindent \textbf{Theorem~\ref{thrm: equiv}:} 
We introduce the Min-WCS problem and prove that, 
for any constant $r \geq 1$,
every $r$-approximation algorithm of the Min-WCS problem 
can be transformed into an $r$-approximation algorithm for 
the Min-Age problem.

\noindent \textbf{Theorem~\ref{thrm: appox}:} 
We solve the Min-WCS problem by combining two feasible schedules. 
Specifically, we propose a deterministic 
4-approximation algorithm and a randomized 2.733-approximation algorithm 
for the Min-WCS problem.

\noindent \textbf{Theorem~\ref{thrm: dp}:}
We give a dynamic-programming-based exact algorithm for the Min-WCS problem. 
The result implies that the Min-Age problem can be solved 
optimally in polynomial time when the number of sender-receiver pairs is a 
constant. The result holds even if there are arbitrarily many messages.

\noindent \textbf{Theorem~\ref{thrm: NP}:}
We show that the Min-Age problem is NP-hard. 

\section{Problem Definition}
The studied problem is first considered by He \textit{et al.}, 
and is referred to as the minimum age scheduling problem with TDMA~\cite{IT}. 
Throughout this paper, we simply refer to this problem as the Min-Age problem.
To make the paper self-contained, 
we rephrase the definition of the Min-Age problem.

\textbf{Inputs:}
We consider $n$ sender-receiver pairs, 
$(s_1, r_1), (s_2, r_2), \cdots, (s_n, r_n)$, 
where $s_i$ and $r_i$ are the sender and receiver of the $i$th 
sender-receiver pair, respectively.
Time is indexed by non-negative integers, and the current time is $T_0$. 
These $n$ sender-receiver pairs share one transmission channel, 
which can transmit one message in one unit of time (hence the name TDMA).
Each sender $s_i$ has a set of messages $\mathcal{M}_i$ 
to be sent to receiver $r_i$.
Our task is to schedule the transmissions of messages in 
$\mathcal{M}_1 \cup \mathcal{M}_2 \cup \cdots \cup \mathcal{M}_n$.

We use $b(M)$ (the birthday of $M$) to indicate 
the time at which message $M$ is generated.
Let $M_i^0$ be the latest message that has been received by $r_i$ by time 
$T_0$.\footnote{Recall that a receiver is defined as a system 
that changes its state upon reception of an update message. 
The system is first assigned a state during the initialization phase. 
Thus, if $r_i$ has not received any message sent from $s_i$, 
$M_i^0$ is the initial information installed on $r_i$ during the initialization phase.}
Thus, $M_i^0 \notin \mathcal{M}_i$.
Let $M_i^j$ be the $j$th oldest message in $\mathcal{M}_i$.
Thus, for all $1 \leq i \leq n$, 
$0 \leq b(M_i^0) < b(M_i^1) < b(M_i^2) < \cdots 
< b(M_i^{|\mathcal{M}_i|}) \leq T_0$.

\textbf{Output and constraints:} 
The goal is to find a schedule $S$ of message transmissions so that the overall 
age of information (to be defined later) is minimized. 
Let $S(M_i^j)$ be the time at which message $M_i^j$ is received by 
$r_i$ under schedule $S$.
Hence, by the channel capacity constraint, 
$S(M_i^j) - 1$ is the time at which the channel starts to send $M^j_i$ 
under schedule $S$.
Let $T = |\mathcal{M}_1|+|\mathcal{M}_2|+\cdots+|\mathcal{M}_n|$ be the 
time needed to send all the messages. 
A feasible schedule $S$ has to satisfy the following constraints.
\begin{enumerate}
\item Due to the channel capacity constraint, 
$S$ is a one-to-one and onto mapping from $\mathcal{M}_1 \cup \mathcal{M}_2 \cup \cdots \cup \mathcal{M}_n$ to $\{T_0+1, T_0+2, \cdots, T_0+ T\}$.
\item Since a message may depend on previous messages, 
the schedule must follow the order of message generation. 
Specifically, for all $1 \leq i \leq n$, 
$S(M_i^1) < S(M_i^2) < \cdots < S(M_i^{|\mathcal{M}_i|})$.
In other words, for each sender-receiver pair, the transmission 
schedule must follow the first-come-first-served (FCFS) discipline.
\end{enumerate}

\textbf{Age:}
Let $lm(S, i, t)$ be the latest message received by receiver $r_i$ 
at or before time $t$ under schedule $S$. 
The \textbf{age} of $r_i$ at time $t$ is the age of $lm(S, i, t)$ at time $t$, 
i.e., $t - b(lm(S, i, t))$. Like~\cite{IT}, we assume that, 
once $r_i$ receives all messages in $\mathcal{M}_i$, 
the age of $r_i$ becomes zero. 
Intuitively, under this assumption, a scheduling algorithm that minimizes
the overall age would have the side benefit that the last 
message of each sender-receiver pair is sent as early as possible (under the
FCFS discipline).
More supporting arguments for this assumption can be found in~\cite{IT}.
Specifically, the age of $r_i$ at time $t$ under schedule $S$, 
$age(S, i, t)$, is defined as follows.
\begin{align*}
&age(S, i, t) = t - b(lm(S, i, t)), &\text{ if } 
lm(S, i, t) \neq M_i^{|\mathcal{M}_i|},\\
&age(S, i, t) = 0,  &\text{ otherwise.}
\end{align*}
Notice that 
$b(M_i^{|\mathcal{M}_i|})$ is not used when evaluating the age of $r_i$.
Moreover, $age(S, i, T_0) = T_0 - b(M_i^0)$ is referred to as the 
\textit{initial age} of receiver $r_i$.
In Section~\ref{sec: conclusion}, 
we will discuss the case 
where the age of $r_i$ does not become zero even if $r_i$ receives 
all messages in $\mathcal{M}_i$.

\textbf{Objective function:}
In the Min-Age problem, the goal is to minimize the 
overall age, which adds up
the ages of all receivers at all time indices.
Specifically, the goal is to find a feasible schedule $S$ that minimizes
\begin{equation*}
Age(S) = \sum_{i=1}^{n}
{\sum_{t = T_0}^{T_0+T}{age(S, i, t)}}.
\end{equation*}

\begin{figure}[t]
    \includegraphics[width=13 cm]{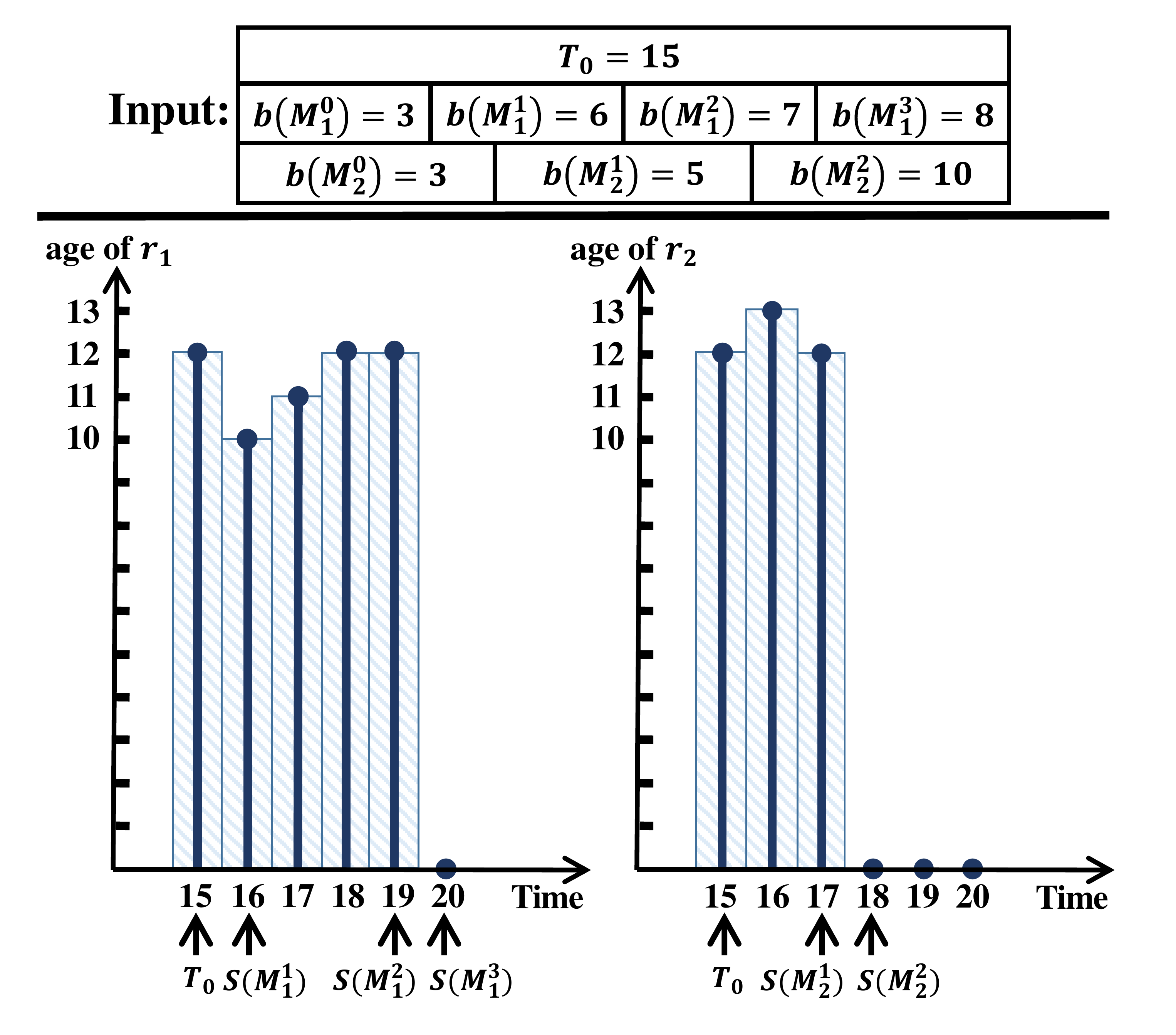}
    \caption{An example of the Min-Age problem.}
    \label{fig: ex_defi}
\end{figure}
\begin{ex}[Min-Age Problem] \label{ex:min-age}
We give an example in~\cite{IT} with our notation.\footnote{The 
example is shown in Fig.~5 in~\cite{IT}.} 
We consider two sender-receiver pairs, where $|\mathcal{M}_1| = 3$ 
and $|\mathcal{M}_2| = 2$. 
Specifically, 
\begin{align*}
&T_0 = 15\\
&b(M_1^0) = 3, b(M_1^1) = 6, b(M_1^2) = 7, b(M_1^3) = 8\\ 
&b(M_2^0) = 3, b(M_2^1) = 5, b(M_2^2) = 10.
\end{align*}
Consider the schedule $S$ shown in Fig.~\ref{fig: ex_defi} with
\begin{align*}
&S(M_1^1) = 16, S(M_1^2) = 19, S(M_1^3) = 20\\ 
&S(M_2^1) = 17, S(M_2^2) = 18. 
\end{align*}
Observe that $S$ is a one-to-one and onto mapping from 
$\mathcal{M}_1 \cup \mathcal{M}_2$ to $\{T_0+1, T_0+2, \cdots, T_0+ T\}$, 
where $T_0 = 15$ and $T = 5$. 
Moreover, $S$ follows the first-come-first-served policy.
Hence, $S$ is a feasible schedule.
$\sum_{t = T_0}^{T_0+T}{age(S, 1, t)} = 
(15-3) + (16-6) + (17-6) + (18-6) + (19-7) + 0 = 57$.
$\sum_{t = T_0}^{T_0+T}{age(S, 2, t)} = 
(15-3) + (16-3) + (17-5) + 0 + 0 + 0 = 37$.
Hence, $Age(S) = 57+37=94$.
\end{ex}
\section{A Corresponding Job Scheduling Problem and Problem Transformation}
\label{sec: job}
In this paper, we cast the Min-Age problem as a job scheduling problem called the 
Min-WCS problem. 
We first give the definition of the Min-WCS problem in Section~\ref{sec: JS}. 
We then show that the Min-Age problem
can be transformed into the Min-WCS problem in Section~\ref{sec: trans}.

\subsection{The Min-WCS Problem}
\label{sec: JS}
We consider a job scheduling problem with 
precedence constraints. That is, the order of job completion has to follow 
a given precedence relation $\rightarrow$. Specifically, 
for any two jobs $J_1$ and $J_2$,
if $J_1 \rightarrow J_2$, then $S(J_1) < S(J_2)$, where 
$S(J)$ is the completion time of job $J$ under schedule $S$.
We consider chain-like precedence constraints. 
Specifically, the set of all jobs is divided into $n_{chain}$ \textbf{job chains}, 
$\mathcal{C}_1, \mathcal{C}_2, \cdots, \mathcal{C}_{n_{chain}}$, 
where $\mathcal{C}_i$ is a chain of $|\mathcal{C}_i|$ jobs, 
$J_i^1 \rightarrow J_i^2 \rightarrow \cdots \rightarrow J_i^{|\mathcal{C}_i|}$.
For any feasible job schedule $S$ and any $1 \leq i \leq n_{chain}$, 
$S(J_i^1) < S(J_i^2) < \cdots < S(J_i^{|\mathcal{C}_i|})$.
Throughout this paper, $J_i^j$ denotes the $j$th job of 
job chain $\mathcal{C}_i$. $J_i^j$ is called a \textbf{leaf job} 
if $j = |\mathcal{C}_i|$; 
otherwise, it is called an \textbf{internal job}. 

We are now ready to define the job scheduling problem considered in
this paper. The input consists of $n_{chain}$ job chains,
where each job $J_i^j$ is associated 
with a non-negative weight $w_i^j$.
The processing time of every job is one unit of time,  
and the system only has one machine, which starts processing jobs at time 0. 
All jobs are non-preemptive.
Hence, the completion time of the last completed job is
$T_{chain} = |\mathcal{C}_1| + |\mathcal{C}_2| + \cdots + |\mathcal{C}_{n_{chain}}|$.  
Since the processing time of each job is one unit of time, 
a feasible schedule is a one-to-one and onto mapping from the set of all jobs 
to $\{1, 2, \cdots, T_{chain}\}$.
The goal is to find a feasible schedule $S$ that minimizes $wcs(S) = wc(S)+cs(S)$,
where $wc(S)$ is the total weighted completion time of all jobs under $S$, 
and $cs(S)$ is the total completion time squared of all leaf jobs under $S$.
Specifically, 
\begin{equation*}
wc(S) = \sum_{\text{All jobs }J_i^j}{(w_i^j \cdot S(J_i^j))},
\end{equation*}
and 
\begin{equation*}
cs(S) = \sum_{\text{All leaf jobs }J_i^{|\mathcal{C}_i|}}
{(S(J_i^{|\mathcal{C}_i|}) \cdot S(J_i^{|\mathcal{C}_i|}))}.
\end{equation*}
In this paper, we refer to this job scheduling problem as 
the Min-WCS problem.

\subsection{Transformation from the Min-Age Problem to the Min-WCS Problem}
\label{sec: trans}
In this subsection, we give a method to solve the Min-Age problem by 
transforming it into the Min-WCS problem. The high-level idea is to 
construct a corresponding job $J_i^j$ for each message $M_i^j \in \mathcal{M}_i$.
Specifically, given a problem instance $I_{age}$ of the Min-Age problem, 
we construct a corresponding instance $I_{job}$ of the Min-WCS problem, where 
\begin{equation}\label{eq: trans1}
n_{chain} = n,
\end{equation}
and
\begin{equation}\label{eq: trans2}
|\mathcal{C}_i| = |\mathcal{M}_i|, \text{ for all } 1 \leq i \leq n.
\end{equation}
The job weight is determined by $T_0$ and $b(M)$.
Specifically, 
\begin{equation}
\label{eq: trans3} 
w_i^j = 2(b(M_i^j) - b(M_i^{j-1})), \text{ if } j \in \{1, 2, \cdots, 
|\mathcal{C}_i|-1\},
\end{equation}
and
\begin{equation}
\label{eq: trans4} 
w_i^{|\mathcal{C}_i|} = 2(T_0 - 0.5 - b(M_i^{|\mathcal{M}_i|-1})).
\end{equation}
Note that, since 
$b(M_i^{|\mathcal{M}_i|-1}) \leq b(M_i^{|\mathcal{M}_i|}) - 1 \leq T_0 - 1$,
all weights are non-negative, and thus this is a valid problem instance 
of the Min-WCS problem. 
Since we have $n_{chain} = n$ and $T_{chain} = T$ in the transformation, 
in what follows, we omit the subscript of 
$n_{chain}$ and $T_{chain}$.

\begin{ex}[The transformation]\label{ex: tran}
Consider the Min-Age problem instance $I_{age}$ in Example~\ref{ex:min-age}.
We transform $I_{age}$ into the following instance $I_{job}$ of the Min-WCS problem.
$I_{job}$ has two jobs chains. The first job chain has three jobs, 
and the second job chain has two jobs.
The weights of the first two jobs in $\mathcal{C}_1$ are 
\begin{equation*}
w_1^1 = 2(b(M_1^1) - b(M_1^0)) = 2(6 - 3) = 6,
\end{equation*}
and
\begin{equation*}
w_1^2 = 2(b(M_1^2) - b(M_1^1)) = 2(7 - 6) = 2.
\end{equation*}
The weight of the last job in $\mathcal{C}_1$ is 
\begin{equation*}
w_1^3 = 2(T_0 - 0.5 - b(M_1^2)) = 2(15 - 0.5 - 7) = 15.
\end{equation*}
Similarly, we have $w_2^1 = 4$ and $w^2_2 = 19$.
Recall that, in Fig.~\ref{fig: ex_defi}, $Age(S) = 94$.
Consider a schedule $S_{job}$ such that
$S_{job}(J^j_i) = S(M^j_i) - T_0$ for all $1 \leq i \leq 2$, 
$1 \leq j \leq |\mathcal{C}_i|$.
We then have 
$wc(S_{job}) = 6 \cdot 1 + 4 \cdot 2 + 19 \cdot 3 + 2 \cdot 4 + 15 \cdot 5 = 154$
and $cs(S_{job}) = 5 \cdot 5 + 3 \cdot 3 = 34$.
Notice that $wcs(S_{job}) = wc(S_{job}) + cs(S_{job}) = 154+34=188 = 2 \cdot Age(S)$.
\end{ex}

\begin{figure}
    \includegraphics[width=13 cm]{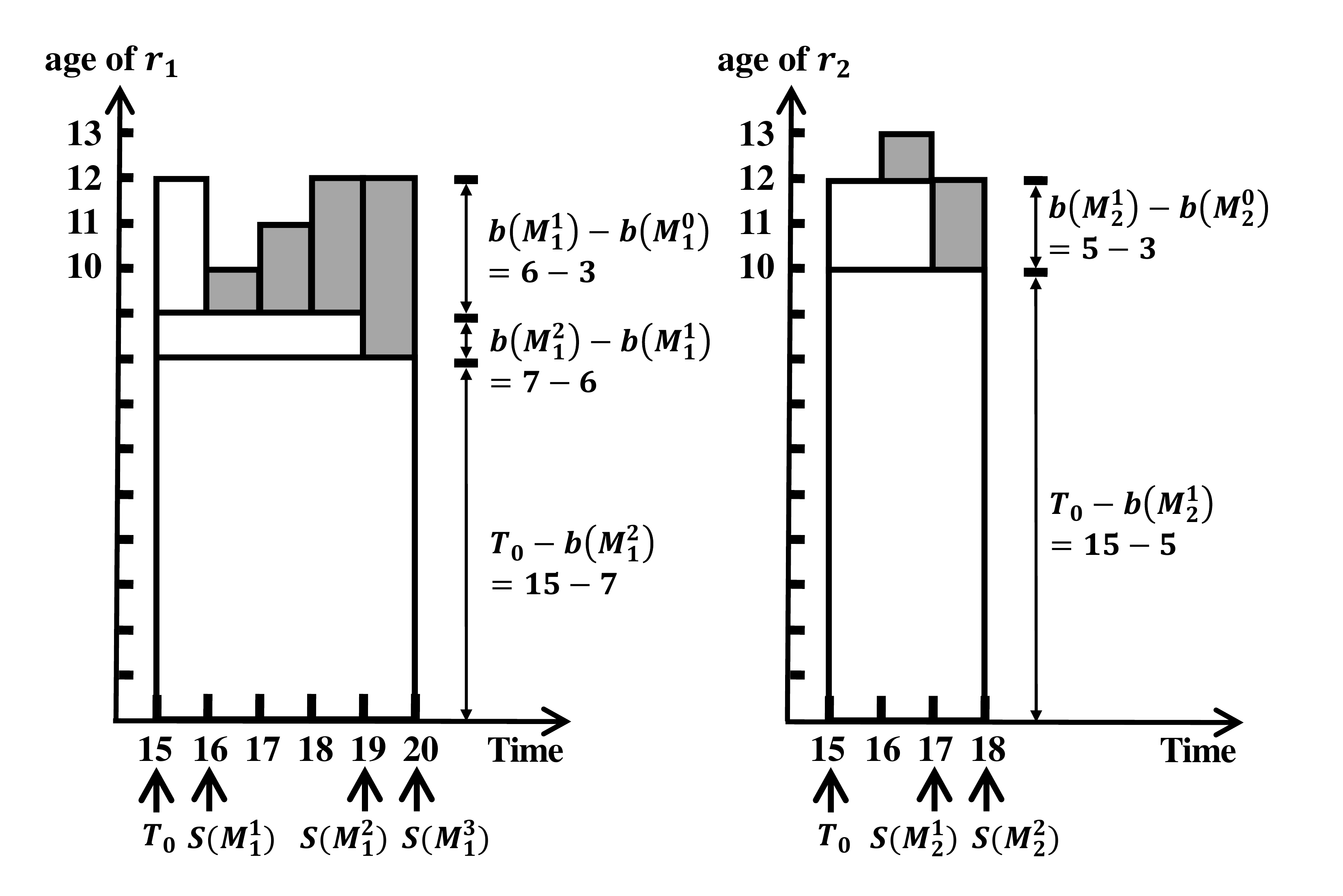}
    \caption{A geometric interpretation of $Age(S)$.}
    \label{fig: ex_trans}
\end{figure}

\textbf{The rationale behind the transformation:}
We give a geometric interpretation of $Age(S)$.\footnote{He \textit{et al.}
also gave a geometric interpretation of $Age(S)$~\cite{IT}.
The geometric interpretation proposed in this paper is different from that in \cite{IT}, 
and our interpretation naturally suggests a transformation into the 
job scheduling problem defined in this paper.}
We use Fig.~\ref{fig: ex_trans} to explain the idea.
Notice that in Fig.~\ref{fig: ex_defi}, $Age(S)$ 
is the total area of rectangles shown in Fig.~\ref{fig: ex_defi}.
In Fig.~\ref{fig: ex_trans}, we divide the overall age of $r_i$ into
white rectangles and gray rectangles.
Since we only consider the total area, we right-shift all rectangles by 0.5 unit.
For $r_i$, there are $|\mathcal{M}_i|$ white rectangles, 
and the width of the $j$th white rectangle is $S(M_i^j) - T_0$.
The height of the $j$th white rectangle is 
$b(M_i^j) - b(M_i^{j-1})$ (if $1 \leq j \leq |\mathcal{M}_i|-1$)
or $T_0 - b(M_i^{|\mathcal{M}_i|-1})$ (if $j = |\mathcal{M}_i|$).
The height can be interpreted as the age reduction after
receiving message $M_i^j$. Note that, after receiving the last message, 
the age becomes zero. Hence, the total height
of the white rectangles should be $T_0 - b(M_i^0)$, 
i.e., the initial age of $r_i$.
Therefore, the height of the bottom white rectangle is 
$T_0 - b(M_i^0) - \sum_{j = 1}^{|\mathcal{M}_i|-1}{(b(M_i^j) - b(M_i^{j-1}))} = 
T_0 - b(M_i^{|\mathcal{M}_i|-1})$.
After considering age reduction, we still need to increase the age by one 
after each unit of time. This is captured by the gray rectangles.
The width of every gray rectangle is one, and the heights of gray 
rectangles are $1, 2, \cdots, S(M_i^{|\mathcal{M}_i|}) - T_0 - 1$. 
Hence, the total area of the gray rectangles is 
$\frac{(S(M_i^{|\mathcal{M}_i|}) - T_0)(S(M_i^{|\mathcal{M}_i|}) - T_0 -1)}{2} 
= \frac{(S(M_i^{|\mathcal{M}_i|}) - T_0)^2}{2} - 
\frac{S(M_i^{|\mathcal{M}_i|}) - T_0}{2}$.
Let $S_{age}$ be any feasible schedule 
of a Min-Age problem instance $I_{age}$. We have 
\begin{align*}
&Age(S_{age})\\
&= \sum_{i=1}^{n}{\sum_{j=1}^{|\mathcal{M}_i|-1}
{(b(M_i^j) - b(M_i^{j-1}))(S_{age}(M_i^j)-T_0)}}\\
&+\sum_{i=1}^{n}
{(T_0 - b(M_i^{|\mathcal{M}_i|-1}))(S_{age}(M_i^{|\mathcal{M}_i|})-T_0)}\\
&+\sum_{i=1}^{n}{(\frac{(S_{age}(M_i^{|\mathcal{M}_i|}) - T_0)^2}{2} - 
\frac{S_{age}(M_i^{|\mathcal{M}_i|}) - T_0}{2})}.
\end{align*}

Let $I_{job}$ be $I_{age}$'s corresponding job scheduling problem instance. Specifically, 
$I_{age}$ and $I_{job}$ satisfy Eq.~\eqref{eq: trans1} to Eq.~\eqref{eq: trans4}.
Let $S_{job}$ be any feasible schedule of $I_{job}$. We have  
\begin{align*}
&wcs(S_{job}) = \sum_{i=1}^{n}{\sum_{j=1}^{|\mathcal{C}_i|-1}
{2(b(M_i^j) - b(M_i^{j-1}))S_{job}(J_i^j)}}\\
&+\sum_{i=1}^{n}
{2(T_0 -\frac{1}{2}- b(M_i^{|\mathcal{M}_i|-1}))S_{job}(J_i^{|\mathcal{C}_i|})}\\
&+\sum_{i=1}^{n}{(S_{job}(J_i^{|\mathcal{C}_i|}))^2}.
\end{align*}
Thus, if $S_{job}(J_i^j) = S_{age}(M_i^j) - T_0$ holds for all $1 \leq i \leq n$, 
$1 \leq j \leq |\mathcal{M}_i|$, we then have $2Age(S_{age}) = wcs(S_{job})$.

The above result then suggests the following method to construct 
a schedule $S_{age}$ for $I_{age}$.
First, obtain a schedule $S_{job}$ 
of the corresponding Min-WCS problem instance $I_{job}$.
We then view $S_{job}(J_i^j)$ as the transmission order of $M_i^j$ in $S_{age}$.
Specifically, we set
$S_{age}(M_i^j) = S_{job}(J_i^j) + T_0$. The following lemma establishes 
the relation between $S_{job}$ and $S_{age}$.
Throughout this paper, we use $I_{age}$ and $I_{job}$ to denote problem instances 
of the Min-Age problem and the Min-WCS problem, respectively.

\begin{lemma}\label{lemma: trans}
Let $S_{age}$ and $S_{job}$ be any two schedules of $I_{age}$ and $I_{job}$, 
respectively. If $I_{age}$ and $I_{job}$ satisfy 
Eq.~\eqref{eq: trans1} to Eq.~\eqref{eq: trans4}, and 
$S_{age}(M_i^j) = S_{job}(J_i^j) + T_0$, for all $1 \leq i \leq n$, 
$1 \leq j \leq |\mathcal{C}_i|$, then 
\begin{enumerate}
\item $S_{age}$ is feasible if and only if $S_{job}$ is feasible.
\item $2Age(S_{age}) = wcs(S_{job})$.
\end{enumerate}
\end{lemma}

\begin{proof}
By the above discussion, we already have $2Age(S_{age}) = wcs(S_{job})$.
Since $S_{age}(M_i^j) = S_{job}(J_i^j) + T_0$ for all $1 \leq i \leq n$, 
$1 \leq j \leq |\mathcal{C}_i|$, $S_{age}$ is a one-to-one and onto mapping 
from $\bigcup_{i=1}^n{\mathcal{M}_i}$ to $\{T_0+1, T_0+2, \cdots, T_0+ T\}$ 
if and only if $S_{age}$ is a one-to-one and onto mapping 
from the set of all jobs to $\{1, 2, \cdots, T\}$.
On the other hand, it is easy to see that $S_{age}$ follows the
first-come-first-served policy for each sender-receiver pair if and only if 
$S_{job}$ follows the chain-like precedence constraint.
Thus, $S_{age}$ is feasible if and only if $S_{job}$ is feasible.
\end{proof}

The next lemma establishes the relation between the optimums of a 
Min-Age problem instance and the corresponding Min-WCS problem instance.
\begin{lemma} \label{lemma: equiv}
Let $S^{*}_{age}$ and $S^{*}_{job}$ be the optimal schedules of $I_{age}$ 
and $I_{job}$, respectively.
If $I_{age}$ and $I_{job}$ satisfy 
Eq.~\eqref{eq: trans1} to Eq.~\eqref{eq: trans4}, 
then $2Age(S^{*}_{age}) = wcs(S^{*}_{job})$.
\end{lemma}

\begin{proof}
Let $S'_{age}$ be a schedule such that 
$S'_{age}(M_i^j) = S^{*}_{job}(J_i^j) + T_0$ 
for all $1 \leq i \leq n$, $1 \leq j \leq |\mathcal{C}_i|$.
Similarly, 
let $S'_{job}$ be a schedule such that 
$S^{*}_{age}(M_i^j) = S'_{job}(J_i^j) + T_0$ 
for all $1 \leq i \leq n$, 
$1 \leq j \leq |\mathcal{C}_i|$.
By Lemma~\ref{lemma: trans}, we have
\begin{equation*}
2Age(S'_{age}) = wcs(S^{*}_{job}) \text{ and } 2Age(S^{*}_{age}) = wcs(S'_{job}).
\end{equation*}
Finally, since
\begin{equation*}
2Age(S^{*}_{age}) \leq 2Age(S'_{age}) = wcs(S^{*}_{job})
\end{equation*}
and
\begin{equation*}
wcs(S^{*}_{job}) \leq wcs(S'_{job}) = 2Age(S^{*}_{age}),
\end{equation*}
we have $wcs(S^{*}_{job}) = 2Age(S^{*}_{age})$.
\end{proof}

\begin{thrm}\label{thrm: equiv}
For any constant $r \geq 1$,
if there exists a polynomial-time $r$-approximation algorithm 
for the Min-WCS problem, 
then there exists a polynomial-time $r$-approximation algorithm 
for the Min-Age problem.
\end{thrm}

\begin{proof}
The $r$-approximation algorithm 
for the Min-Age problem proceeds as follows. 
First, given a problem instance $I_{age}$ of the Min-Age problem, 
the algorithm constructs a corresponding instance $I_{job}$ 
of the Min-WCS problem by the aforementioned transformation.
Obviously, the transformation can be done in polynomial time. 
We then apply the $r$-approximation algorithm for 
the Min-WCS problem on $I_{job}$ to get a schedule $S_{job}$. 
We construct a schedule $S_{age}$ for $I_{age}$ by
setting $S_{age}(M_i^j) = S_{job}(J_i^j) + T_0$ for all $1 \leq i \leq n$, 
$1 \leq j \leq |\mathcal{M}_i|$.
By Lemmas~\ref{lemma: trans} and \ref{lemma: equiv}, $S_{age}$ is feasible and 
$Age(S_{age}) = \frac{wcs(S_{job})}{2} \leq 
r \cdot \frac{wcs(S^{*}_{job})}{2} = 
r \cdot Age(S^{*}_{age})$.
\end{proof} 

\section{Approximation Algorithms for the Min-WCS Problem}\label{sec: approx}
By Theorem~\ref{thrm: equiv}, to solve the Min-Age problem, it suffices to solve 
the Min-WCS problem. Notice that the objective function of the 
Min-WCS problem is 
the sum of two functions, $wc$ and $cs$. 
When the objective function becomes $wc$ (respectively, $cs$), 
we refer to the problem as the Min-WC problem 
(respectively, the Min-CS problem).
Both the Min-WC problem and the Min-CS problem 
can be solved optimally in polynomial time. 
Given an instance of the Min-WCS problem, 
the high-level idea of our algorithm is to first
solve the corresponding instances of the Min-WC problem and the Min-CS problem.
Throughout this paper, we use $S^{*}_{wc}$ (respectively, $S^{*}_{cs}$) 
to denote the optimal schedule of the Min-WC problem
(respectively, the Min-CS problem).
We then \textit{interleave} $S^{*}_{cs}$ with $S^{*}_{wc}$ 
to approximate the Min-WCS problem.
We first discuss the solutions of the Min-WC problem 
and the Min-CS problem in Section~\ref{subsec: WC-CS}.
We then present our algorithm for the Min-WCS problem in 
Section~\ref{subsec: Min-WCS}.

\subsection{Algorithms for the Min-WC Problem and the Min-CS Problem}
\label{subsec: WC-CS}
\subsubsection{The Min-WC Problem}
The Min-WC problem is a special case of the minimum total 
weighted completion time
scheduling problem subject to precedence constraints, 
which has been studied over many 
years~\cite{Sidney, LAWLER197875, Hu}. 
When the precedence constraints are chain-like,
the problem can be solved in polynomial time~\cite{LAWLER197875, Hu}.
Recall that, in our problem, the processing time of every job is one.
The algorithm for the Min-WC problem proceeds as follows.
For each job $J_i^j$, define the job's priority $\rho_i^j$ as 
$\max_{k: j \leq k \leq |\mathcal{C}_i|}
{\frac{w_i^j+w_i^{j+1} + \cdots + w_i^k}{k-j+1}}$.
To minimize the total weighted completion time, 
the machine should first process the job with the highest priority.
We still need to follow the precedence constraints.
Hence, to determine the next processing job, we only consider the first 
unprocessed job in each job chain, and we choose the one that 
has the highest priority.
Algorithm~\ref{algo: WC} summarizes the pseudocode.

\begin{algorithm}[t]
\begin{small}
\caption{An Algorithm for the Min-WC Problem}
\label{algo: WC}
\For {$t \gets 1$ \KwTo 
$|\mathcal{C}_1|+|\mathcal{C}_2|+\cdots+|\mathcal{C}_n|$} {
  $\mathcal{U} \gets$ the set of the first unscheduled job in each job chain\\
  $J \gets \argmax_{J_i^j \in \mathcal{U}}{\rho_i^j}$\\    
  $S^*_{wc}(J) \gets t$\\
}
\end{small}
\end{algorithm}

\begin{lemma}[Lawler~\cite{LAWLER197875}]
Algorithm~\ref{algo: WC} solves the Min-WC problem optimally in polynomial time.
\end{lemma}

\begin{ex}[Algorithm~\ref{algo: WC}]\label{ex: Min-WC}
Consider the problem instance in Example~\ref{ex: tran}.
We have $\rho_1^1 = \max{(\frac{6}{1}, \frac{6+2}{2}, \frac{6+2+15}{3})} 
= \frac{23}{3}$ and 
$\rho_2^1 = \max{(\frac{4}{1}, \frac{4+19}{2})} = \frac{23}{2}$.
Since $\rho_2^1 > \rho_1^1$, Algorithm~\ref{algo: WC} first 
schedules $J_2^1$ and sets $S^*_{wc}(J_2^1) = 1$.
The job completion order under 
$S^*_{wc}$ is $J_2^1, J_2^2, J_1^1, J_1^2, J_1^3$.
\end{ex}

\subsubsection{The Min-CS Problem}
By a simple interchange argument, it is easy to see that the shortest job chain 
should be completed first in the Min-CS problem.
Algorithm~\ref{algo: CS} summarizes the pseudocode.
We have the following lemma.
\begin{lemma}
Algorithm~\ref{algo: CS} solves the Min-CS problem optimally in polynomial time.
\end{lemma}

\begin{algorithm}[t]
\begin{small}
\caption{An Algorithm for the Min-CS Problem}
\label{algo: CS}
$t \gets 1$\\
$\mathcal{U} \gets \{1, 2, \cdots, n\}$\\
\While {$\mathcal{U} \neq \varnothing$}{
    $i^* \gets \argmin_{i \in \mathcal{U}}{|\mathcal{C}_i|}$\\
    $\mathcal{U} \gets \mathcal{U} \setminus \{i^*\}$\\
    \For {$j \gets 1$ \KwTo $|\mathcal{C}_{i^*}|$} {
        $S^*_{cs}(J_{i^*}^j) \gets t$\\
        $t \gets t+1$\\ 
    	}
}
\end{small}
\end{algorithm}

\begin{ex}[Algorithm~\ref{algo: CS}]
\label{ex: Min-CS}
Consider the problem instance in Example~\ref{ex: tran}.
Since $|\mathcal{C}_2| < |\mathcal{C}_1|$, 
the job completion order under $S^*_{cs}$ is 
$J_2^1, J_2^2, J_1^1, J_1^2, J_1^3$.
\end{ex}

Observe that in Example~\ref{ex: Min-WC} and Example~\ref{ex: Min-CS},
$S^*_{wc} = S^*_{cs}$. It is easy to see that $S^*_{wc}$ and $S^*_{cs}$ 
are thus optimal schedules of the Min-WCS problem. Therefore, the 
optimal message transmission order in Example~\ref{ex:min-age} 
is $M_2^1, M_2^2, M_1^1, M_1^2, M_1^3$, and the optimal overall age is 
$(12+13+14+12+12)+(12+11)=86$.
\begin{prop}
Let $I_{job}$ be any instance of the Min-WCS problem, and 
let $S^*_{wc}$ and $S^*_{cs}$ be the optimal schedules of 
the corresponding instances of the Min-WC problem and the Min-CS problem, 
respectively. If $S^*_{wc} = S^*_{cs}$, then $S^*_{wc}$ and $S^*_{cs}$ are
optimal schedules of $I_{job}$.
\end{prop}

\subsection{Interleaving $S^*_{wc}$ and $S^*_{cs}$ Randomly: 
A Randomized Approximation Algorithm for the Min-WCS Problem}
\label{subsec: Min-WCS}
While $S^*_{wc}$ and $S^*_{cs}$ solve the Min-WC problem and the Min-CS problem, 
respectively, neither $S^*_{wc}$ nor $S^*_{cs}$ can approximate 
the Min-WCS problem well. Specifically, we have the following results, whose proof can be 
found in the appendix.
\begin{prop}
\label{prop: wc}
The approximation ratio of Algorithm~\ref{algo: WC} for the Min-WCS problem
is $\Omega(n)$.
\end{prop}
\begin{prop}
\label{prop: cs}
The approximation ratio of Algorithm~\ref{algo: CS} for the Min-WCS problem
is $\Omega(n)$.
\end{prop}

Despite the above negative results, we will show that interleaving 
$S^*_{wc}$ and $S^*_{cs}$ gives an $O(1)$-approximation algorithm for 
the Min-WCS problem.
A critical observation of the Min-WC problem (respectively, the Min-CS problem) 
is that, if we multiply the optimal scheduled completion time $S^*_{wc}(J)$ 
(respectively, $S^*_{cs}(J)$) of every job $J$
by a factor $c > 1$ (i.e., we \textit{delay} the 
optimal schedule by a multiplicative \textit{delay factor} of $c$), 
then the total weighted completion time (respectively, 
the total completion time squared of all leaf jobs) 
is increased by a multiplicative factor of $c$ (respectively, $c^2$).
This immediately suggests the following deterministic 4-approximation algorithm: 
For each job $J$, set $S^{int}_{cs}(J) = 2S^*_{cs}(J)-1$.
Hence, $S^{int}_{cs}$ is a delayed version of $S^*_{cs}$ with a delay factor less than 
two\footnote{Although different jobs have different delay factors, 
every job has a delay factor less than two.}, 
and the time period $[2k-1,2k]$ is idle for any integer $k \geq 1$.
We call such an idle time period an \textit{idle time slot}.
Moreover, define the \textit{finish time} of an idle time slot $[t-1,t]$ as $t$.
Consider another schedule $S^{int}_{wc}$ obtained by setting
$S^{int}_{wc}(J) = 2S^*_{wc}(J)$ for each job $J$.
Hence, $S^{int}_{wc}$ is a delayed version of $S^*_{wc}$ with a delay factor of two. 
We can view $S^{int}_{wc}$ as a schedule obtained by inserting 
jobs one by one following the order specified in $S^*_{wc}$ 
to the idle time slots in $S^{int}_{cs}$.
For each job $J$, set 
$S'_{wcs}(J) = \min{\{S^{int}_{wc}(J), S^{int}_{cs}(J)\}}$.
We will show that $S'_{wcs}$ satisfies the precedence constraints. 
Finally, we remove the idle time slots in $S'_{wcs}$ to 
obtain the final schedule $S_{wcs}$.
We then have
\begin{equation*}
wc(S_{wcs}) \leq wc(S^{int}_{wc}) \leq 2 \cdot wc(S^*_{wc})
\end{equation*}
and 
\begin{equation*}
cs(S_{wcs}) \leq cs(S^{int}_{cs}) \leq 2^2 \cdot cs(S^*_{cs}).
\end{equation*}
Thus,
\begin{equation*}
wcs(S_{wcs}) = wc(S_{wcs})+cs(S_{wcs})
             \leq 4(wc(S^*_{wc})+cs(S^*_{cs})).
\end{equation*} 
Since $wc(S^*_{wc})+cs(S^*_{cs})$ is a lower bound of the optimum of the 
Min-WCS problem, $S_{wcs}$ is a 4-approximation solution. 

\begin{algorithm}[t]
\begin{small}
\caption{An Algorithm for the Min-WCS Problem with Parameter $p$}
\label{algo: WCS}
$S^*_{wc} \gets$ the schedule obtained by Algorithm~\ref{algo: WC}\\ 
$S^*_{cs} \gets$ the schedule obtained by Algorithm~\ref{algo: CS}\\ 
$S^{int}_{cs} \gets S^*_{cs}$\\
$T \gets |\mathcal{C}_1|+|\mathcal{C}_2|+\cdots+|\mathcal{C}_n|$\\

\For{$i \gets 1$ \KwTo $T-1$}{
    $X_i$ is set to 1 with probability $p$ and is set to 0 with probability $1-p$\\
    \If{$X_i = 1$}{
         \ForAll{Job $J$ such that $S^*_{cs}(J) > i$}{
             $S^{int}_{cs}(J) \gets S^{int}_{cs}(J)+1$
         }
    }
}


\For{$i \gets 1$ \KwTo $T$}{
    $J \gets$ the $i$th completed job under $S^*_{wc}$\\
    $S^{int}_{wc}(J) \gets$ 
    the finish time of the $i$th idle time slot in $S^{int}_{cs}$\\ 
}

\ForAll{$J \in \mathcal{C}_1 \cup \mathcal{C}_2 \cup \cdots \cup \mathcal{C}_n$}
{
    $S'_{wcs}(J) \gets \min{\{S^{int}_{cs}(J), S^{int}_{wc}(J)\}}$\\
}

\For{$i \gets 1$ \KwTo $T$}{
    $J \gets$ the $i$th completed job under $S'_{wcs}$\\
    $S_{wcs}(J) = i$\\
}
\end{small}
\end{algorithm}

In hindsight, we first insert idle time slots to $S^*_{cs}$ and
then insert jobs to the idle time slots following the order specified in 
$S^{*}_{wc}$. To improve the algorithm, we insert idle time slots to 
$S^*_{cs}$ randomly.
Specifically, let $p$ be a number in $[0,1]$.
Initially, $S^{int}_{cs} = S^*_{cs}$.
For every two jobs $J_1$ and $J_2$ 
that are processed contiguously in $S^{*}_{cs}$ 
(i.e., $|S^{*}_{cs}(J_2) - S^{*}_{cs}(J_1)|=1$),
we insert an idle time slot between $S^{int}_{cs}(J_1)$ and 
$S^{int}_{cs}(J_2)$ with probability $p$.
Notice that, in $S^{int}_{cs}$, we never insert two or more 
contiguous idle time slots, 
which is a critical property that will be used in the analysis.
Algorithm~\ref{algo: WCS} summarizes the pseudocode.
Observe that this randomized algorithm degenerates to Algorithm~\ref{algo: CS}
when $p = 0$, and this randomized algorithm degenerates to the aforementioned 
deterministic 4-approximation algorithm when $p = 1$.

\begin{figure}[t]
    \includegraphics[width=13 cm]{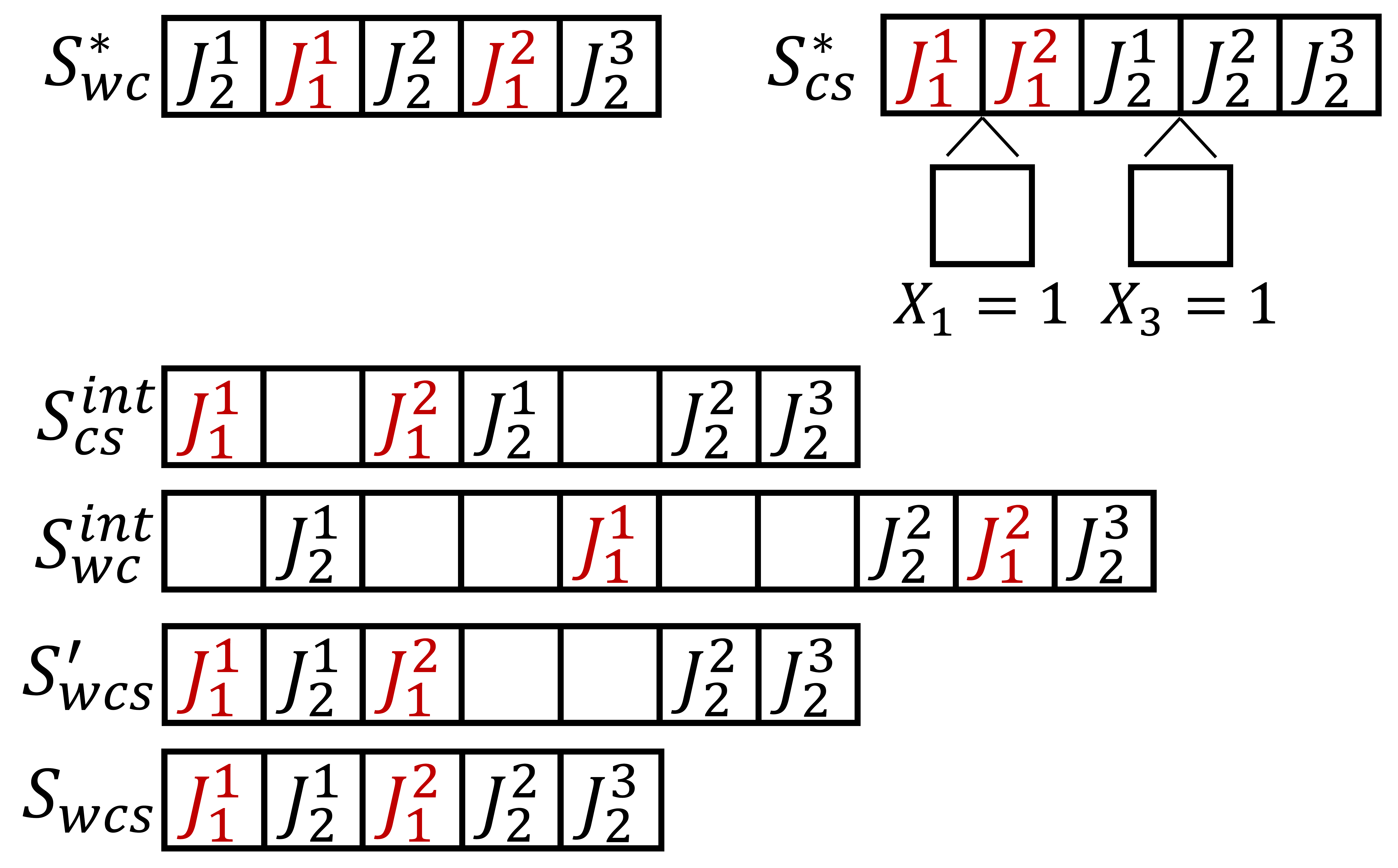}
    \caption{An example of Algorithm~\ref{algo: WCS}.}
    \label{fig: ex_algo}
\end{figure}

\begin{ex}[Algorithm~\ref{algo: WCS}]
Consider a Min-WCS problem instance with two job chains where 
$|\mathcal{C}_1| = 2$ and $|\mathcal{C}_2| = 3$.
Hence, the job completion order under $S^*_{cs}$ is 
$J^1_1, J^2_1, J^1_2, J^2_2, J^3_2$.
Assume that the job completion order under $S^*_{wc}$ is 
$J^1_2, J^1_1, J^2_2, J^2_1, J^3_2$.
Assume $X_1 = X_3 = 1, X_2 = X_4 = 0$.
$S^{int}_{wc}, S^{int}_{cs}, S'_{wcs}$ and $S_{wcs}$ are shown in 
Fig.~\ref{fig: ex_algo}.
\end{ex}

Since $S^{int}_{wc}$ and $S^{int}_{cs}$ do not overlap, 
we never execute two jobs at the same time in $S'_{wcs}$.
Thus, to prove that $S_{wcs}$ is feasible, 
it remains to prove the following lemma.
\begin{lemma} \label{lemma: feasible}
$S'_{wcs}$ follows the precedence constraints.
\end{lemma}

\begin{proof}
For the sake of contradiction, assume that there are two jobs 
$J^a_i$ and $J^b_i$ 
from the same job chain such that $S'_{wcs}(J^a_i) > S'_{wcs}(J^b_i)$ but 
$a < b$. We first consider the case 
where $S'_{wcs}(J^a_i) = S^{int}_{wc}(J^a_i)$.
Hence, we must have $S'_{wcs}(J^b_i) = S^{int}_{cs}(J^b_i)$ 
(otherwise, $S^{int}_{wc}$ and $S^*_{wc}$ would violate 
the precedence constraints).
Since $S^{int}_{cs}$ follows the precedence constraints, 
$S^{int}_{cs}(J^b_i) > S^{int}_{cs}(J^a_i)$. 
Finally, since $S'_{wcs}(J^a_i) > S'_{wcs}(J^b_i) = S^{int}_{cs}(J^b_i)$, 
we have $S'_{wcs}(J^a_i) > S^{int}_{cs}(J^a_i)$, 
which contradicts to the definition of $S'_{wcs}$.
The case where $S'_{wcs}(J^a_i) = S^{int}_{cs}(J^a_i)$ 
can be proved in a similar way.
\end{proof}

Throughout this paper, we use $\mathbf{E}[X]$ to denote the 
expected value of $X$. The following theorem expresses 
the approximation ratio of Algorithm~\ref{algo: WCS} as a function of $p$.
\begin{thrm}\label{thrm: appox}
$\frac{\mathbf{E}[wcs(S_{wcs})]}{wc(S^*_{wc}) + cs(S^*_{cs})} 
\leq \max{\{1+\frac{1}{p}, 1+3p\}}$.
\end{thrm}
\begin{proof}
It is sufficient to prove
\begin{equation}\label{eq: WC}
\mathbf{E}[wc(S^{int}_{wc})] \leq (1+\frac{1}{p})wc(S^*_{wc})
\end{equation}
and
\begin{equation}\label{eq: CS}
\mathbf{E}[cs(S^{int}_{cs})] \leq (1+3p)cs(S^*_{cs}).
\end{equation}
This is because, by Eq.~\eqref{eq: WC} and Eq.~\eqref{eq: CS}, we have
\begin{align*}
&\mathbf{E}[wcs(S_{wcs})] \\
&= \mathbf{E}[wc(S_{wcs})+cs(S_{wcs})]\\
&=\mathbf{E}[wc(S_{wcs})]+\mathbf{E}[cs(S_{wcs})]\\
&\leq \mathbf{E}[wc(S^{int}_{wc})] + \mathbf{E}[cs(S^{int}_{cs})]\\ 
&\leq (1+1/p)wc(S^*_{wc}) + (1+3p)cs(S^*_{cs})\\
&\leq \max{\{1+1/p, 1+3p\}} \cdot (wc(S^*_{wc}) + cs(S^*_{cs})).
\end{align*}

Let $J^i_{wc}$ be the $i$th completed job under $S^{*}_{wc}$.
Hence, $S^{*}_{wc}(J^i_{wc}) = i$.
Let $Y_i$ be the completion time of $J^i_{wc}$ under $S^{int}_{wc}$, 
i.e., $Y_i = S^{int}_{wc}(J^i_{wc})$.
Similarly, let $J^i_{cs}$ be the $i$th completed job under $S^{*}_{cs}$, 
and let $Z_i = S^{int}_{cs}(J^i_{cs})$.

To prove Eq.~\eqref{eq: WC},
it suffices to show that 
$\mathbf{E}[Y_i] \leq i(1+1/p)$ holds 
for all $1 \leq i \leq T$.  
Let $G_p$ be a random variable indicating 
the number of trials required to get the first success where 
the probability of success in each independent trial is $p$.
By the setting of $S^{int}_{wc}$, 
we have $Y_1 = \min{\{G_p+1, T+1\}}$.
Moreover, for all $2 \leq i \leq T$, 
$Y_i = \min{\{Y_{i-1} + G_p + 1, T+i\}} 
= \min{\{i(G_p+1), T+i\}}$. 
Hence,
\begin{equation*}
\mathbf{E}[Y_i] \leq \mathbf{E}[i(G_p+1)]
=i(1+\mathbf{E}[G_p]) = i(1+1/p).
\end{equation*} 

To prove Eq.~\eqref{eq: CS},
it suffices to prove that
$\mathbf{E}[Z_i^2] \leq (1+3p)i^2$ holds 
for all $1 \leq i \leq T$.  
Let $B_p$ be a random variable such that 
$B_p = 1$ with probability $p$ and $B_p = 0$ with probability $1-p$.
By the setting of $S^{int}_{cs}$, 
we have $Z_1 = 1$, and $Z_2 = Z_1+B_p+1 = 2+B_p$. 
In general, for all $2 \leq i \leq T$, we have 
$Z_i = Z_{i-1}+B_p+1 = i+(i-1)B_p$.
Hence, 
\begin{align*}
&\mathbf{E}[Z_i^2] = \mathbf{E}[(i+(i-1)B_p)^2]\\
&=\mathbf{E}[i^2+2i(i-1)B_p+(i-1)^2B_p^2]\\
&=i^2+2i(i-1)\mathbf{E}[B_p]+(i-1)^2\mathbf{E}[B_p^2]\\
&=i^2+2i(i-1)p+(i-1)^2p \\
&= (1+3p)i^2 -4pi+p\\
&\leq (1+3p)i^2 \text{ (since } i \geq 1\text{)}. \qedhere
\end{align*} 
\end{proof}

When $p = \frac{1}{\sqrt{3}}$, $1+\frac{1}{p} = 1+3p = 1+\sqrt{3}$.
Note that, $\frac{1}{\sqrt{3}} \approx 0.57735$
and $1+\sqrt{3} \approx 2.73205$.
\begin{coro}\label{coro: appox}
Algorithm~\ref{algo: WCS} is a randomized 
2.733-approximation algorithm for the Min-WCS problem when $p = 0.57735$. 
\end{coro}

\begin{coro}\label{coro: deter-appox}
Algorithm~\ref{algo: WCS} is a deterministic
4-approximation algorithm for the Min-WCS problem when $p = 1$. 
\end{coro}

\section{An Exact Algorithm for the Min-WCS Problem}
\label{sec: DP}
Next, we solve the Min-WCS problem by dynamic programming.
The objective function $wcs$ can be stated as follows.
\begin{equation*}
wcs(S) = \sum_{i=1}^{n}{\sum_{j = 1}^{|\mathcal{C}_i|}{f_i^j(S(J_i^j))}},
\end{equation*}
where $\displaystyle f^j_i(t) = w^j_i \cdot t \text{ if } j \neq |\mathcal{C}_i|$,
and $\displaystyle f^j_i(t) = w^j_i \cdot t + t^2 \text{ if } j = |\mathcal{C}_i|$.
In other words, $f^j_i(t)$ is the cost incurred by job $J^j_i$ if 
$J^j_i$ is completed at time $t$, and $wcs(S)$ is simply the total cost 
incurred by all jobs under schedule $S$.

In the dynamic program, we consider subproblems of the Min-WCS problem where
job chains can be executed partially. 
Specifically, for each job chain $\mathcal{C}_i$, we only need to 
schedule the first $\mathbf{L}[i]$ 
jobs $J_i^1, J_i^2, \cdots, J_i^{\mathbf{L}[i]}$, 
where $\mathbf{L}[i] \leq |\mathcal{C}_i|$ and $\mathbf{L}[i]$ may be zero.
When $\mathbf{L}[i] = 0$, we do not need to schedule any job in $\mathcal{C}_i$.
More formally, let $\mathbf{L}$ be any vector of length $n$ such that, 
for all $1 \leq i \leq n$, 
the $i$th element of 
$\mathbf{L}$, denoted by $\mathbf{L}[i]$, 
is in $\{0, 1, 2, \cdots, |\mathcal{C}_i|\}$.
In this paper, for any vector $\mathbf{V}$, we use $\mathbf{V}[i]$ to denote
the $i$th element of $\mathbf{V}$.
Define 
\begin{equation*}
wcs_{\mathbf{L}}(S) = \sum_{i=1}^{n}{\sum_{j = 1}^{\mathbf{L}[i]}{f_i^j(S(J_i^j))}}.
\end{equation*}
Let $\mathcal{J}(\mathbf{L})$ be the set of jobs whose costs are considered 
in $wcs_{\mathbf{L}}$. 
Let $\mathbf{MinWCS}(\mathbf{L})$ be a subproblem of the Min-WCS problem where 
we only need to schedule jobs in $\mathcal{J}(\mathbf{L})$ 
with objective function $wcs_{\mathbf{L}}$. 
Hence, the Min-WCS problem is equivalent to 
$\mathbf{MinWCS}(\mathbf{L})$ when 
$\mathbf{L}[i] = |\mathcal{C}_i|$ for all $1 \leq i \leq n$.
Let $S^*_{\mathbf{L}}$ be the optimal schedule of $\mathbf{MinWCS}(\mathbf{L})$.

Observe that the last completed job under $S^*_{\mathbf{L}}$ must be
$J^{\mathbf{L}[k]}_k$ for some $1 \leq k \leq n$.
To find $S^*_{\mathbf{L}}$, 
we try every possible last completed job $J^{\mathbf{L}[k]}_k$
and consider the subproblem obtained by removing $J^{\mathbf{L}[k]}_k$ from 
$\mathbf{MinWCS}(\mathbf{L})$.
Define $\mathbf{L}_{-k}$ as a vector of length $n$ such that
$\mathbf{L}_{-k}[j] = \mathbf{L}[j]$ if $j \neq k$
and $\mathbf{L}_{-k}[k] = \mathbf{L}[k]-1$.
Thus, the subproblem $\mathbf{MinWCS}(\mathbf{L}_{-k})$ can
be obtained by removing $J^{\mathbf{L}[k]}_k$ from 
$\mathbf{MinWCS}(\mathbf{L})$.
We then compute $S^*_{\mathbf{L}}$ based on 
$S^*_{\mathbf{L}_{-1}}, S^*_{\mathbf{L}_{-2}}, \cdots, S^*_{\mathbf{L}_{-n}}$.
Specifically, define $S^*_{\mathbf{L}_{-k}} \oplus J^{\mathbf{L}[k]}_k$ 
as the schedule obtained by setting the
completion time of job $J^{\mathbf{L}[k]}_k$ 
to $\sum_{1 \leqslant i \leqslant n}{\mathbf{\mathbf{L}}[i]}$ and 
setting the completion time of every job $J'$ in $\mathcal{J}(\mathbf{L}_{-k})$ to 
$S^*_{\mathbf{L}_{-k}}(J')$. 
Since $wcs_{\mathbf{L}}$ simply adds up the cost incurred by each job in 
$\mathcal{J}(\mathbf{L})$,
it is easy to see that
\begin{equation*}
S^*_{\mathbf{L}} = \argmin_{S \in \{ S^*_{\mathbf{L}_{-k}} \oplus J^{\mathbf{L}[k]}_k |
 1 \leq k \leq n, \mathbf{L}[k] \geq 1\}}
{wcs_{\mathbf{L}}(S)}.
\end{equation*}

Given the values of $\sum_{1 \leqslant i \leqslant n}{\mathbf{L}[i]}$ and 
$wcs_{\mathbf{L}_{-k}}(S^*_{\mathbf{L}_{-k}})$, 
$wcs_{\mathbf{L}}(S^*_{\mathbf{L}_{-k}} \oplus J^{\mathbf{L}[k]}_k)$ 
can be computed in $O(1)$ time. Hence, given all the required 
$S^*_{\mathbf{L}_{-k}}$ and $wcs_{\mathbf{L}_{-k}}(S^*_{\mathbf{L}_{-k}})$,
$\min_{S \in \{ S^*_{\mathbf{L}_{-k}} \oplus J^{\mathbf{L}[k]}_k |
1 \leq k \leq n, \mathbf{L}[k] \geq 1\}}
{wcs_{\mathbf{L}}(S)}$ can be computed in $O(n)=O(T)$ time. 
We can then construct $S^*_{\mathbf{L}}$ 
in $O(\sum_{1 \leqslant i \leqslant n}{\mathbf{L}[i]}) = O(T)$ time.
The base case where $\mathbf{L}[i] = 0$ for all $1 \leq i \leq n$ is trivial, 
and we can solve the dynamic program by a bottom-up approach.
Since there are $\prod_{i=1}^{n}{(|\mathcal{C}_i|+1)}$ subproblems, 
we have the following results. 

\begin{thrm}\label{thrm: dp}
The Min-WCS problem can be solved optimally in 
$O(T\prod_{i=1}^{n}{(|\mathcal{C}_i|+1)})$ time.
\end{thrm}

Thus, the Min-WCS problem can be solved optimally in polynomial time when 
the number of job chains is a constant (i.e., $n = O(1)$), 
even if there are arbitrarily many jobs.

\begin{coro}\label{coro: dp}
The Min-WCS problem can be solved optimally in polynomial time if $n=O(1)$.
This is true even if there are arbitrarily many jobs.
\end{coro}

\section{NP-Hardness of the Min-Age Problem}\label{sec: NP}
In this section, we prove that the Min-Age problem is NP-hard.
He \textit{et al.} proved that a certain generalization of the Min-Age problem is 
NP-hard~\cite{IT}. However, this result does not preclude the possibility of 
solving the Min-Age problem optimally in polynomial time.
Specifically, He \textit{et al.} studied a generalization 
of the Min-Age problem where senders in the same \textit{candidate group} 
can send messages simultaneously. 
The list of candidate groups
are either explicitly specified in the inputs 
or can be derived from an interference model based on SINR.
This generalization greatly increases the complexity of the scheduling problem. 
In fact, He \textit{et al.} proved that, 
even if every sender $s_i$ has only one message to be scheduled, 
(i.e., $|\mathcal{M}_i| = 1$ for all $1 \leq i \leq n$), 
the generalization is still NP-hard~\cite{IT}. 
However, this special case can be solved optimally in polynomial time for 
the Min-Age problem.\footnote{To see this, consider the corresponding Min-WCS 
problem instance, 
where every job chain has only one job.
Then, for any two feasible schedules $S_1$ and $S_2$ 
of the Min-WCS problem instance, $cs(S_1) = cs(S_2)$.
Hence, the corresponding Min-WCS problem instance 
can be solved optimally in polynomial time by Algorithm~\ref{algo: WC}.
This implies that the Min-Age problem can be solved optimally in 
polynomial time if $|\mathcal{M}_i| = 1$ for all $1 \leq i \leq n$.}
On the other hand, we transform the Min-Age problem into 
the Min-WCS problem, where the processing time of every job is one, 
and the precedence constraints are chain-like. 
Given such a simple setting, one may suspect that 
the Min-WCS problem is in P, and thus the Min-Age problem is in P as well.
Nevertheless, in this paper, we prove that
the Min-Age problem is NP-hard.
Hence, unless P $=$ NP, 
the best polynomial-time algorithm for the Min-Age problem 
is an approximation algorithm.

{The big picture of the proof is the following.
We first consider a special case of the Min-WCS problem 
called the Constrained-Min-WCS problem. 
We prove that, if the Constrained-Min-WCS problem is NP-hard, 
then the Min-Age problem is NP-hard (Section~\ref{sec: NP1}). 
We then prove that the Constrained-Min-WCS problem is NP-hard by a reduction 
from the 3-Partition problem~\cite{Garey:1990:CIG:574848} 
(Section~\ref{sec: NP2}).

\subsection{Reduction from the Constrained-Min-WCS Problem 
to the Min-Age Problem}
\label{sec: NP1}
Recall that, in Section~\ref{sec: trans}, we give a transformation from the 
Min-Age problem to the Min-WCS problem. 
Notice that we can reverse the transformation 
when job weights satisfy certain properties. 
\begin{defi}
The \textbf{Constrained-Min-WCS problem} is a special case 
of the Min-WCS problem, where the weight of every internal 
job is an even number greater than zero, 
and the weight of every leaf job is an odd number greater than zero.
\end{defi}

\begin{lemma}
\label{lemma: NP0}
If the Constrained-Min-WCS problem is NP-hard, 
then the Min-Age problem is NP-hard.
\end{lemma}

\begin{proof}
Given an instance $I_{job}$ of the Constrained-Min-WCS problem, we construct a 
corresponding instance $I_{age}$ of the Min-Age problem.
In $I_{age}$, the number of sender-receiver pairs, $n$, 
is equal to the number of 
job chains in $I_{job}$, $n_{chain}$.
$|\mathcal{M}_i| = |\mathcal{C}_i|$, for all $1 \leq i \leq n$.
For each $1 \leq i \leq n$, define 
$W_i = \frac{w_i^{|\mathcal{C}_i|}+1}{2}+
\sum_{j=1}^{|\mathcal{C}_i|-1}{\frac{w_i^j}{2}}$.
Furthermore, define $W_{max} = \max_{1 \leq i \leq n}{W_i}$.
We set $b(M_i^j)$ and $T_0$ as follows.
\begin{align*}
&b(M_i^0) = W_{max}-W_i, &\forall 1 \leq i \leq n,\\
&b(M_i^j) = b(M_i^{j-1}) + \frac{w_i^j}{2}, 
&\forall 1 \leq i \leq n, 1 \leq j \leq |\mathcal{M}_i| - 1\\
&b(M_i^{|\mathcal{M}_i|}) = W_{max}, &\forall 1 \leq i \leq n,\\
&T_0 = W_{max}. &
\end{align*}
Obviously, this transformation can be done in polynomial time.

We first prove that $I_{age}$ is a valid Min-Age problem instance.
Since the weight of every leaf job is odd, 
and the weight of every internal job is even, 
$W_i$ is a non-negative integer for all $1 \leq i \leq n$.
It is then easy to see that, for every message $M$ in $I_{age}$, 
$b(M)$ is an integer.
Since the weight of an internal job is at least two,
we have $b(M_i^0) < b(M_i^1) < ... < b(M_i^{|\mathcal{M}_i|-1})$.
To prove $b(M_i^{|\mathcal{M}_i|-1}) < b(M_i^{|\mathcal{M}_i|})$, 
observe that under our setting, 
\begin{equation*}
b(M_i^{|\mathcal{M}_i|-1}) 
= \sum_{j=1}^{|\mathcal{C}_i|-1}{\frac{w_i^j}{2}}+W_{max}-W_i
= W_{max}-\frac{w_i^{|\mathcal{C}_i|}+1}{2}
< W_{max} = b(M_i^{|\mathcal{M}_i|}).
\end{equation*}

To complete the proof, it suffices to prove that 
$I_{age}$ and $I_{job}$ satisfy Eq.~\eqref{eq: trans3} and 
Eq.~\eqref{eq: trans4}. 
The proof then follows from Lemmas~\ref{lemma: equiv} and \ref{lemma: trans}.
When $j \in \{1, 2, \cdots, |\mathcal{C}_i|-1\}$, since
\begin{equation*}
2(b(M_i^j) - b(M_i^{j-1})) = 2(\frac{w_i^j}{2}) = w_i^j,
\end{equation*}
Eq.~\eqref{eq: trans3} holds.
Finally, Eq.~\eqref{eq: trans4} holds because
\begin{equation*}
2(T_0 - 0.5 - b(M_i^{|\mathcal{M}_i|-1}))
=2(W_{max}-0.5-(W_{max}-\frac{w_i^{|\mathcal{C}_i|}+1}{2}))=w_i^{|\mathcal{C}_i|}.
\end{equation*}
\end{proof}

\subsection{Reduction from the 3-Partition Problem to the 
Constrained-Min-WCS Problem}
\label{sec: NP2}  
We first give the definition of the 3-Partition problem.
\begin{defi}
Given a set $L$ of $3m$ positive integers, 
$a_1, a_2, \cdots, a_{3m}$, and a positive integer $B$  
such that $\sum_{a \in L}{a} = mB$ and $\frac{B}{4} < a < \frac{B}{2}$ 
for all $a \in L$, the \textbf{3-Partition problem} asks for a partition of $L$
into $m$ subsets of $L$, $P_1, P_2, \cdots, P_m$, 
such that $\sum_{a \in P_i}{a} = B$, for all $1 \leq i \leq m$.    
Note that, since $\frac{B}{4} < a < \frac{B}{2}$ 
for all $a \in L$, $|P_i| = 3$ for all $1 \leq i \leq m$.
\end{defi}
The 3-Partition problem is \textit{unary NP-hard}, i.e.,
the 3-Partition problem is NP-hard even if integers are encoded in 
unary~\cite{Garey:1990:CIG:574848}.
If integers are encoded in unary, 
then the space required to represent an integer $n$ is $\Theta(n)$, 
instead of $\Theta(\log n)$ in binary encoding.
Moreover, even if all the integers in $L$ and $B$ are even numbers greater than 
zero, the 3-Partition problem is still unary NP-hard.
To see this, given any instance $I_{3P}$ with input 
$L = \{a_1, a_2, \cdots, a_{3m}\}$ 
and $B$ of the 3-Partition problem in unary encoding, 
we can construct an instance $I^{even}_{3P}$ of the 3-Partition problem 
in unary encoding with input $L' = \{2a_1, 2a_2, \cdots, 2a_{3m}\}$ 
and $B' = 2B$ in polynomial time. 
Since $a > \frac{B}{4} > 0$ for all $a \in L$, 
all integers in $L'$ and $B'$ are even numbers greater than zero. 
It is easy to see that $I_{3P}$ has a feasible solution if and only if 
$I^{even}_{3P}$ has a feasible solution. 

In the following proofs, we will consider an extension of the Min-WCS problem 
called the \textbf{NonUni-Min-WCS problem}.
The only difference between the NonUni-Min-WCS problem and the Min-WCS problem 
is that in the NonUni-Min-WCS problem, 
different jobs may have different processing times.
Like the Min-WCS problem, there is only one machine 
and the schedule is non-preemptive in the NonUni-Min-WCS problem. 

The proof of the NP-hardness of the Constrained-Min-WCS problem consists 
of two steps. In the first step, 
we prove that the 3-Partition problem in unary encoding
can be reduced to the NonUni-Min-WCS problem in polynomial time even if the 
processing time is in unary encoding (Lemma~\ref{lemma: NP1}). 
We then prove that the Constrained-Min-WCS problem is NP-hard by a reduction 
from the NonUni-Min-WCS problem where the processing time is in unary encoding 
(Lemma~\ref{lemma: NP2}).
To avoid ambiguity, we use $wcs$ and $\overbar{wcs}$ to denote 
the objective functions of the NonUni-Min-WCS problem and 
the Constrained-Min-WCS problem, respectively

In the following lemmas and proofs, when we refer to the NonUni-Min-WCS problem 
or the Constrained-Min-WCS problem, 
we mean the \textit{decision version} of these problems. 
Specifically, let $P$ be a minimization problem with objective function $f$.
In the decision version of $P$, 
we are given one additional input $Q$,
and we have to decide whether the problem has a solution $S$ 
such that $f(S) \leq Q$ and $S$ satisfies all the constraints of $P$.
A solution $S$ is \textbf{valid} for the decision version of $P$, 
if it satisfies all the constraints of $P$.
Moreover, a solution $S$ is \textbf{feasible} for the decision version of $P$, 
if it is valid and $f(S) \leq Q$.
Again, to avoid ambiguity, we use $Q$ and $\overbar{Q}$ to denote 
the additional inputs of the NonUni-Min-WCS problem and 
the Constrained-Min-WCS problem, respectively.

\begin{lemma}
\label{lemma: NP1}
The NonUni-Min-WCS problem is NP-hard, 
even if the processing time is in unary encoding and 
all the following properties are satisfied:
\begin{itemize}
\item[P1:] All the weights and processing times 
of internal jobs are even numbers greater than zero.
\item[P2:] All the weights of leaf jobs are odd numbers 
greater than zero.
\item[P3:] The processing time of every leaf job is one.
\item[P4:] Every job chain has at most two jobs.
\end{itemize}
\end{lemma}

\begin{proof}
Given an instance $I_{3P}$ of the 3-Partition problem 
with inputs $L = \{a_1, a_2, \cdots, a_{3m}\}$ and $B$,
we construct an instance $I_{job}$ of the NonUni-Min-WCS problem.
As discussed previously, we can assume that all integers 
in $L$ and $B$ are even numbers greater than zero, 
and these numbers are in unary encoding.
Before we proceed to the reduction, we stress that we only need to consider 
schedules that has no idle time slots, 
i.e., the machine is always processing some job unless
all jobs are completed. This is because, given a schedule with idle time slots,
we can always find in polynomial time a better schedule that has no idle time slots.
The high-level idea of the reduction is the following:
\begin{enumerate}
\item For each $a_i \in L$, we create a corresponding job chain $\mathcal{C}_i$ 
consisting of two jobs. The processing time of the first job in $\mathcal{C}_i$ 
is $a_i$. We thus refer to the first job in $\mathcal{C}_i$ as an 
\textbf{$a$-job}. 
The second job (leaf job) is a \textbf{dummy job}, 
whose purpose is to make sure that
the completion time of the $a$-job is not counted in $cs(S)$. 
We will set the weights of dummy jobs to 
some relatively small numbers so that it is safe to assume that all 
these $3m$ dummy jobs are completed lastly in a feasible schedule
(Lemma~\ref{lemma: dummy}).
\item We then create $m-1$ job chains 
$\mathcal{C}^*_1, \cdots, \mathcal{C}^*_{m-1}$.
Each job chain $\mathcal{C}^*_i$ has only one job, 
which, by definition, is a leaf job. 
We refer to these jobs as the \textbf{separating jobs}.
Since dummy jobs are completed lastly,
jobs completed between separating jobs are $a$-jobs, 
and thus $m-1$ separating jobs partition the set of $a$-jobs into $m$ subsets. 
The processing time of every separating job is one.
We will fine-tune the job weights and $Q$, 
so that if $I_{job}$ has a feasible schedule,
then the completion times of these $m-1$
separating jobs under the feasible schedule must be 
$(B+1), 2(B+1), \cdots, (m-1)(B+1)$ (Lemma~\ref{lemma: separating}).
Hence, the partition yielded by separating jobs 
then corresponds to a feasible partition of $I_{3P}$.
\end{enumerate}

We are now ready to construct the corresponding instance $I_{job}$ of 
the NonUni-Min-WCS problem. Let
\begin{equation*}
r = 10mB(B+1).
\end{equation*}
For each $a_i \in L$, we create a job chain $\mathcal{C}_i$ consisting 
of two jobs.
The processing times of the first job ($a$-job) and the second job (dummy job)
in $\mathcal{C}_i$ are $a_i$ and one, respectively.
The weights of the $a$-job and the dummy job in $\mathcal{C}_i$ are $ra_i$ 
and one, respectively.
We use $J_i^a$, $w_i^a$, and $p_i^a$ to denote the 
$a$-job in $\mathcal{C}_i$, its weight, and its processing time, respectively.
Similarly, we use $J_i^d$, $w_i^d$, and $p_i^d$ to denote the dummy-job 
in $\mathcal{C}_i$, its weight, and its processing time, respectively.
Finally, we create $m-1$ single-job job chains, 
$\mathcal{C}^*_1, \mathcal{C}^*_2, \cdots, \mathcal{C}^*_{m-1}$. 
The processing time and the weight of the job in 
$\mathcal{C}^*_{i}$ are one and $r-2i(B+1)+1$, respectively.
We use $J_i$, $w_i$, and $p_i$ to denote the job (separating job) in 
$\mathcal{C}^*_i$, its weight, and its processing time, respectively.
We thus have
\begin{align*}
&\mathcal{C}_i = J_i^a \rightarrow J_i^d, & 1 \leq i \leq 3m,\\
&w_i^a = ra_i, w_i^d = 1 &1 \leq i \leq 3m,\\
&p_i^a = a_i, p_i^d = 1 &1 \leq i \leq 3m,\\
&\mathcal{C}^*_i = J_i, & 1 \leq i \leq m-1,\\
&w_i = r-2i(B+1)+1 &1 \leq i \leq m-1,\\
&p_i = 1, &1 \leq i \leq m-1.
\end{align*}

Finally, we set
\begin{align*}
&Q = \sum_{i = 1}^{3m}{(w_i^a\sum_{j = 1}^{i}{p_j^a})} 
+ \sum_{i = 1}^{m-1}{r(m-i)B}
+ \sum_{i = 1}^{m-1}{w_ii(B+1)}\\
&+\sum_{i=1}^{3m}{w_i^d(m(B+1)-1+i)}\\
&+\sum_{i = 1}^{3m}{(m(B+1)-1+i)^2} +\sum_{i = 1}^{m-1}{(i(B+1))^2}.
\end{align*}
It is easy to see that $I_{job}$ satisfies all the four properties in 
Lemma~\ref{lemma: NP1}. Furthermore, even if the processing time 
in $I_{job}$ is in 
unary encoding, the reduction can be done in polynomial time with respect 
to the size of $I_{3P}$ in unary encoding.

Our goal is to complete separating job $J_i$ at time $i(B+1)$ 
in a feasible schedule for all $1 \leq i \leq m-1$.
Throughout this proof, for all $1 \leq i \leq m-1$, 
define $\Delta^{S}_i = S(J_i) - i(B+1)$. 
In other words, $\Delta^{S}_i$ is the difference between the scheduled 
completion time of $J_i$ under $S$ and the goal.
We stress that $\Delta^{S}_i$ may be negative.
We have the following lemmas, whose proofs can be found in the appendix.

\begin{lemma}\label{lemma: as-job only-}
For any valid schedule $S$ of $I_{job}$, 
if all the dummy jobs are completed lastly in $S$ and 
$\Delta^{S}_i = 0$ holds for all $1 \leq i \leq m-1$, 
then the total weighted completion time of all $a$-jobs and separating jobs
under $S$ is $\sum_{i = 1}^{3m}{(w_i^a\sum_{j = 1}^{i}{p_j^a})} 
+ \sum_{i = 1}^{m-1}{r(m-i)B} 
+ \sum_{i = 1}^{m-1}{w_ii(B+1)}$. 
\end{lemma}

\begin{lemma} \label{lemma: dummy}
Let $S$ be any feasible schedule of $I_{job}$ such that some dummy job 
is completed before some non-dummy job (i.e., $a$-job or separating job). 
There is a feasible schedule $S'$ of $I_{job}$ 
such that all dummy jobs are completed lastly in $S'$. 
Moreover, given $S$, $S'$ can be found in polynomial time.
\end{lemma}

\begin{lemma} \label{lemma: separating}
Let $S$ be any feasible schedule of $I_{job}$ such that 
1) all dummy jobs are completed lastly in $S$, and 2) $S$ has no idle time slots.
For every separating job $J_i$, $S(J_i) = i(B+1)$ (i.e., $\Delta^{S}_i = 0$).
\end{lemma}

We first prove that if $I_{3P}$ has a feasible partition, then 
$I_{job}$ has a feasible schedule $S$. 
Let $\{P_1, P_2, \cdots, P_m\}$ be a feasible partition of $L$ in $I_{3P}$ 
such that $\sum_{a \in P_i}{a} = B$ and $|P_i| = 3$, for all $1 \leq i \leq m$.
The construction of $S$ is as follows.
First, the $3m$ dummy jobs are completed lastly.
Thus, $S$ satisfies precedence constraints.
Since the total processing time of all the $a$-jobs and the separating jobs is 
$mB+(m-1) = m(B+1)-1$, the first dummy job can be completed at $m(B+1)$.
Specifically, $S(J^d_i) = m(B+1)+i-1$, for all $1 \leq i \leq 3m$.
Separating job $J_i$ is completed at time $S(J_i) = i(B+1)$, 
for all $1 \leq i \leq m-1$. 
Hence, $\Delta^S_i  = 0$ for all $1 \leq i \leq m-1$.
Let $a_{x(i)}, a_{y(i)}$, and $a_{z(i)}$ be the three elements in $P_i$, 
for all $1 \leq i \leq m$.
Hence, $J^a_{x(i)}$, $J^a_{y(i)}$, and $J^a_{z(i)}$ are the corresponding 
$a$-jobs of $a_{x(i)}, a_{y(i)}$, and $a_{z(i)}$, respectively.
For all $1 \leq i \leq m$, we then set 
$S(J^a_{x(i)}) = (i-1)(B+1)+p^a_{x(i)}$, 
$S(J^a_{y(i)}) = S(J^a_{x(i)}) + p^a_{y(i)}$, 
and $S(J^a_{z(i)}) = S(J^a_{y(i)}) + p^a_{z(i)} = (i-1)(B+1)+B$.
It is easy to see that we never process two jobs simultaneously under $S$.
Moreover, $cs(S) = \sum_{i = 1}^{3m}{(m(B+1)-1+i)^2} 
+ \sum_{i = 1}^{m-1}{(i(B+1))^2}$, 
and the total weighted completion time of dummy jobs
is $\sum_{i=1}^{3m}{w^d_i(m(B+1)-1+i)}$. 
Thus, by Lemma~\ref{lemma: as-job only-}, $wcs(S) = Q$.

Finally, we prove that if $I_{job}$ has a feasible schedule, 
then $I_{3P}$ has a feasible partition. 
Given a feasible schedule $S$ of $I_{job}$,
by Lemma~\ref{lemma: dummy}, we can find a feasible schedule $S'$ of $I_{job}$ 
in polynomial time such that
1) all the dummy jobs are completed lastly in $S'$, 
and 2) $S'$ has no idle time slots. 
By Lemma~\ref{lemma: separating},
in $S'$, the separating jobs partition the $a$-jobs into $m$ sets,
$A_1, A_2, \cdots, A_m$, each of which has a total processing time of $B$.
Since $\frac{B}{4} < p^a_i < \frac{B}{2}$ for all $1 \leq i \leq 3m$, 
each $A_i$ has three jobs, denoted by
$J^a_{x(i)}, J^a_{y(i)}$ and $J^a_{z(i)}$.
Set $P_i = \{a_{x(i)}, a_{y(i)}, a_{z(i)}\}$ for all $1 \leq i \leq m$.
$\{P_1, P_2, \cdots, P_m\}$ is then a feasible partition 
of $I_{3P}$.
\end{proof}

Finally, we prove that the Constrained-Min-WCS problem is NP-hard by 
a reduction from the NonUni-Min-WCS problem 
where the processing time is in unary encoding.
In the reduction, for each job $J$ in the NonUni-Min-WCS problem instance, 
if its processing time, $p$, is greater than one, 
we create a job chain of length $p$ in the corresponding 
Constrained-Min-WCS problem instance.
\begin{lemma}
\label{lemma: NP2}
The Constrained-Min-WCS problem is NP-hard.
\end{lemma}

\begin{proof}
Given an instance $I_{job}$ of the NonUni-Min-WCS problem, 
we will construct a corresponding instance $\overbar{I}_{job}$ of 
the Constrained-Min-WCS problem.
By Lemma~\ref{lemma: NP1}, we can assume that
the processing time of $I_{job}$ is in unary encoding, 
and $I_{job}$ satisfies the properties specified in Lemma~\ref{lemma: NP1}.
In this proof, we use $J_i^j$ and $\overbar{J}_i^j$ 
to denote the $j$th job in the $i$th job chain in $I_{job}$ and 
$\overbar{I}_{job}$, respectively. 
Similarly, we use $w_i^j$ and $\overbar{w}_i^j$ 
to denote the weights of $J_i^j$ and $\overbar{J}_i^j$, respectively.

The high-level idea is to replace a job of processing time $p > 1$ with 
a series of $p$ jobs, each of which has unit processing time.
For each job chain $\mathcal{C}_i$ in $I_{job}$, 
we create a job chain $\overbar{\mathcal{C}}_i$.
By P4, $|\mathcal{C}_i| = 1$ or $2$. 
Let $n$ be the number of job chains in $I_{job}$.
Let $n'$ be the number of job chains that have two jobs in $I_{job}$.
Without loss of generality, 
assume that each of the first $n'$ job chains in $I_{job}$ has two jobs, 
and each of the last $n-n'$ job chains in $I_{job}$ has exactly one job. 
The construction of $\overbar{\mathcal{C}}_i$ 
is divided into the following two cases.

\noindent Case 1: $|\mathcal{C}_i| = 2$, i.e., $i \in \{1, 2, \cdots, n'\}$.
Let $p_i^1$ be the processing time of the first job in $\mathcal{C}_i$.
By P1, $p_i^1$ is at least two.
By P3, the processing time of the second job in $\mathcal{C}_i$ is one. 
In $\overbar{I}_{job}$, we create a corresponding job chain
$\overbar{\mathcal{C}}_i$ that has $p_i^1+1$ jobs.
In $\overbar{\mathcal{C}}_i$, the first $p_i^1$ jobs 
simulate the first job in $\mathcal{C}_i$, 
and the last job simulates the last job in $\mathcal{C}_i$.
The weights of the first $p_i^1-1$ jobs in 
$\overbar{\mathcal{C}}_i$ are two. 
$\overbar{w}_i^{p_i^1} = w_i^1 + 2$ and 
$\overbar{w}_i^{p_i^1+1} = w_i^2 + 2$.
This step can be done in polynomial time 
(with respect to the size of $I_{job}$) because the processing time 
of the NonUni-Min-WCS problem is encoded in unary.

\noindent Case 2: $|\mathcal{C}_i| = 1$, 
i.e., $i \in \{n'+1, n'+2, \cdots, n\}$. 
In $\overbar{I}_{job}$, we create a corresponding job chain
$\overbar{\mathcal{C}}_i$ that has only one job $\overbar{J}_i^1$, 
and we set $\overbar{w}_i^1$ to $w_i^1+2$. 

Finally, we set $\overbar{Q} = Q 
+ 2(1 + 2 + \cdots + \sum_{i=1}^{n}{|\overbar{\mathcal{C}}_i|})$.
It is easy to see that this reduction can be done in polynomial time
(with respect to the size of $I_{job}$ where the processing time is 
in unary encoding).

We next prove that $\overbar{I}_{job}$ is a valid Constrained-Min-WCS problem.
In the reduction, for each leaf job $J$ with weight $w$ in $I_{job}$, 
we create a leaf job with weight $w+2$ in $\overbar{I}_{job}$. 
By P2, all the weights of leaf jobs in $\overbar{I}_{job}$ are odd numbers 
greater than zero. For each internal job $J$ with weight $w$ in
$I_{job}$, we create internal jobs with weights $w+2$ or two 
in $\overbar{I}_{job}$. By P1, all the weights of internal jobs in 
$\overbar{I}_{job}$ are even numbers greater than zero. Hence,
$\overbar{I}_{job}$ is a valid Constrained-Min-WCS problem.

It is then sufficient to show that $I_{job}$ has a feasible schedule 
$S$ if and only if $\overbar{I}_{job}$ 
has a feasible schedule $\overbar{S}$.
We first prove the ``only if'' direction. 
The construction of $\overbar{S}$ is straightforward.
Let $J_i^{leaf}$ and $\overbar{J}_i^{leaf}$ be the 
leaf jobs in $\mathcal{C}_i$ and $\overbar{\mathcal{C}}_i$, respectively.
For each leaf job $\overbar{J}_i^{leaf}$ in $\overbar{I}_{job}$, 
$\overbar{S}(\overbar{J}_i^{leaf}) = S(J_i^{leaf})$.
For each internal job $\overbar{J}_i^{p_i^1}$ 
in $\overbar{I}_{job}$, 
$\overbar{S}(\overbar{J}_i^{p_i^1}) = S(J_i^1)$. 
Moreover, 
$\overbar{S}(\overbar{J}_i^{p_i^1})
=\overbar{S}(\overbar{J}_i^{p_i^1-1})+1 
=\overbar{S}(\overbar{J}_i^{p_i^1-2})+2 
=\cdots = \overbar{S}(\overbar{J}_i^1) +(p^1_i-1)$.
Since $S$ is a valid schedule for $I_{job}$, 
$\overbar{S}$ is a valid schedule for $\overbar{I}_{job}$.
It is then sufficient to show that 
$\overbar{wcs}(\overbar{S}) \leq \overbar{Q}$.
Observe that $S$ and $\overbar{S}$ have the same total 
completion time squared of all leaf jobs.
Let $w_i^{leaf}$ and $\overbar{w}_i^{leaf}$ be the 
weights of $J_i^{leaf}$ and $\overbar{J}_i^{leaf}$, respectively.
Compared to $S$, $\overbar{S}$ increases the total weighted 
completion time by 
\begin{align*}
&\sum_{i=1}^{n}
{\overbar{S}(\overbar{J}_i^{leaf})(\overbar{w}_i^{leaf}-w_i^{leaf})}
+\sum_{i=1}^{n'}
{\overbar{S}(\overbar{J}_i^{p_i^1})(\overbar{w}_i^{p_i^1}-w_i^1)}
+\sum_{i=1}^{n'}
{\sum_{j=1}^{p_i^1 -1}{2\overbar{S}(\overbar{J}_i^j)}}\\
&=\sum_{i=1}^{n}
{2\overbar{S}(\overbar{J}_i^{leaf})}
+\sum_{i=1}^{n'}
{2\overbar{S}(\overbar{J}_i^{p_i^1})}
+\sum_{i=1}^{n'}
{\sum_{j=1}^{p_i^1 -1}{2\overbar{S}(\overbar{J}_i^j)}}\\
&= 2(1 + 2 + \cdots + \sum_{i=1}^{n}{|\overbar{\mathcal{C}}_i|}), 
\end{align*}
where the last equality holds since $\overbar{S}$ is a one-to-one and 
onto mapping from the set of all jobs in $\overbar{I}_{job}$ to 
$\{1, 2, \cdots, \sum_{i=1}^{n}{|\overbar{\mathcal{C}}_i|}\}$.
Since $wcs(S) \leq Q$, 
$\overbar{wcs}(\overbar{S}) \leq Q + 2(1 + 2 + \cdots + 
\sum_{i=1}^{n}{|\overbar{\mathcal{C}}_i|}) = \overbar{Q}$.

Finally, we prove the ``if'' direction.
We first state the following lemma, whose proof can be found in the appendix.
\begin{lemma}\label{lemma: conti}
Let $\overbar{S}$ be any feasible schedule of $\overbar{I}_{job}$. 
There exists a feasible schedule $\overbar{S'}$ of $\overbar{I}_{job}$ 
such that $\overbar{S'}(\overbar{J}_i^{p_i^1}) 
= \overbar{S'}(\overbar{J}_i^1) + (p_i^1 - 1)$ for all $1 \leq i \leq n'$.
In other words,  
$\overbar{J}_i^1, \overbar{J}_i^2, \cdots, \overbar{J}_i^{p_i^1}$ 
are processed contiguously under $\overbar{S'}$ for all $1 \leq i \leq n'$.
Moreover, given $\overbar{S}$, $\overbar{S'}$ can be 
found in polynomial time.
\end{lemma}
By Lemma~\ref{lemma: conti}, given a feasible schedule $\overbar{S}$ for $\overbar{I}_{job}$,
we can find in polynomial time 
a feasible schedule $\overbar{S'}$ for $\overbar{I}_{job}$ such that 
$\overbar{J}_i^1, \overbar{J}_i^2, \cdots, \overbar{J}_i^{p_i^1}$ 
are processed contiguously under $\overbar{S'}$ for all $1 \leq i \leq n'$.
We can then view these $p_i^1$ jobs as a single job with process time $p_i^1$.
Specifically, for each leaf job $J_i^{leaf}$ in $I_{job}$, 
we set $S(J_i^{leaf}) = \overbar{S'}(\overbar{J}_i^{leaf})$. 
For each internal job $J_i^1$ in $I_{job}$, we set 
$S(J_i^1) = \overbar{S'}(\overbar{J}_i^{p_i^1})$.
Since $\overbar{S'}$ is a valid schedule for $\overbar{I}_{job}$,
$S$ is a valid schedule for $I_{job}$.
Finally, by the same argument in the proof of the ``only if'' direction,
$wcs(S) \leq Q$.
\end{proof}

By Lemma~\ref{lemma: NP0} and Lemma~\ref{lemma: NP2}, we have the following 
theorem.
\begin{thrm} \label{thrm: NP}
The Min-Age problem is NP-hard.
\end{thrm}

\begin{figure}[t]
    \includegraphics[width=13 cm]{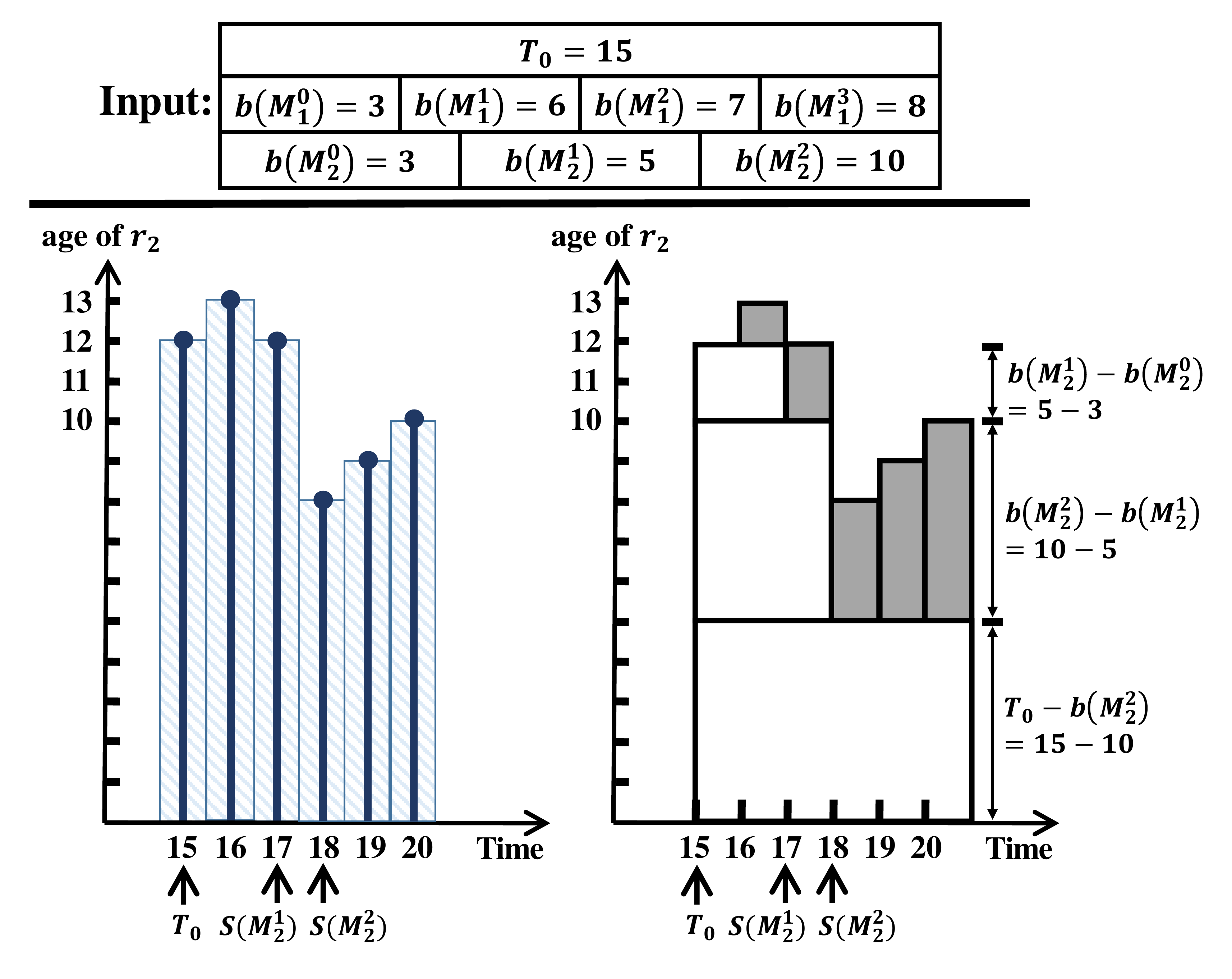}
    \caption{A geometric interpretation of a special receiver's age.}
    \label{fig: discussion}
\end{figure}

\section{Concluding Remarks}\label{sec: conclusion}
In this paper, we assume that the age of a receiver 
$r_i$ becomes zero once $r_i$ receives all messages in $\mathcal{M}_i$. 
One of the rationales behind the design is to make the scheduling 
algorithm transmit the last message for each sender-receiver pair 
as early as possible. 
Nevertheless, we can solve the problem even if the age of some receiver $r_i$ 
is not set to zero after receiving all messages in $\mathcal{M}_i$. 
We call such a receiver a \textbf{special} receiver.
Hence, for a special receiver $r_i$, 
its age is always the age of the most recently received message. 
To solve the Min-Age problem with special receivers, 
we adjust the geometric interpretation given 
in Section~\ref{sec: trans} accordingly. 
An example of the new geometric interpretation is shown in Fig.~\ref{fig: discussion}, 
where $r_2$ is a special receiver.
Recall that $T = |\mathcal{M}_1|+\cdots+|\mathcal{M}_n|$.
Compared to Fig.~\ref{fig: ex_trans}, we have
three critical observations for each special receiver $r_i$: 
\begin{enumerate}
\item The number of white rectangles is increased by one, 
and the area of the bottom white rectangle is
$(T+1)(T_0-b(M_i^{|\mathcal{M}_i|}))$, 
which is fixed regardless of the schedule.
\item The height of the $|\mathcal{M}_i|$th white rectangle becomes 
$b(M_i^{|\mathcal{M}_i|})-b(M_i^{|\mathcal{M}_i|-1})$. 
Therefore, we need to modify the job weight setting in the transformation
accordingly.
\item The total area of the gray rectangles is
$1+2+\cdots+T$, which is fixed regardless of the schedule.
\end{enumerate}
We thus update the objective function of the Min-WCS problem accordingly. 
Specifically, $wcs(S) = wc(S)+cs'(S)+ C$,
where $C$ is a non-negative number specified in the input. 
Moreover, $cs'(S) = \sum_{i=1}^{n_{chain}}
{(I_i\cdot S(J_i^{|\mathcal{C}_i|}) \cdot S(J_i^{|\mathcal{C}_i|}))}$, 
where $I_i$ is an input that can be zero or one.
In the problem transformation, if $r_i$ is a special receiver, 
we set $I_i = 0$. Otherwise, we set $I_i = 1$. 
Finally, we use $C$ to capture the total fixed rectangle area for the special receivers.
To minimize $cs'(S)$, all job chains $\mathcal{C}_i$ with $I_i = 0$ are completed 
lastly in the schedule, and all job chains $\mathcal{C}_i$ with $I_i = 1$ are scheduled 
by Algorithm~\ref{algo: CS}. Hence, we can still compute two optimal schedules that 
minimize $wc$ and $cs'$, respectively, and then 
apply Algorithm~\ref{algo: WCS} to approximate the modified Min-WCS problem. 
Since the constant $C$ in $wcs$ is non-negative, 
the approximation ratio cannot be worse than that in Theorem~\ref{thrm: appox}. 

\bibliographystyle{IEEEtran}

\begin{thebibliography}{10}
\providecommand{\url}[1]{#1}
\csname url@samestyle\endcsname
\providecommand{\newblock}{\relax}
\providecommand{\bibinfo}[2]{#2}
\providecommand{\BIBentrySTDinterwordspacing}{\spaceskip=0pt\relax}
\providecommand{\BIBentryALTinterwordstretchfactor}{4}
\providecommand{\BIBentryALTinterwordspacing}{\spaceskip=\fontdimen2\font plus
\BIBentryALTinterwordstretchfactor\fontdimen3\font minus
  \fontdimen4\font\relax}
\providecommand{\BIBforeignlanguage}[2]{{%
\expandafter\ifx\csname l@#1\endcsname\relax
\typeout{** WARNING: IEEEtran.bst: No hyphenation pattern has been}%
\typeout{** loaded for the language `#1'. Using the pattern for}%
\typeout{** the default language instead.}%
\else
\language=\csname l@#1\endcsname
\fi
#2}}
\providecommand{\BIBdecl}{\relax}
\BIBdecl

\bibitem{conf}
T.-W. Kuo, ``Minimum age tdma scheduling,'' in \emph{IEEE INFOCOM}, 2019.

\bibitem{Yu:1999:SWC:316188.316219}
H.~Yu, L.~Breslau, and S.~Shenker, ``A scalable web cache consistency
  architecture,'' in \emph{ACM SIGCOMM}, 1999.

\bibitem{5307471}
P.~Papadimitratos, A.~D.~L. Fortelle, K.~Evenssen, R.~Brignolo, and S.~Cosenza,
  ``Vehicular communication systems: Enabling technologies, applications, and
  future outlook on intelligent transportation,'' \emph{IEEE Communications
  Magazine}, vol.~47, no.~11, pp. 84--95, November 2009.

\bibitem{1275297}
M.~Xiong and K.~Ramamritham, ``Deriving deadlines and periods for real-time
  update transactions,'' \emph{IEEE Transactions on Computers}, vol.~53, no.~5,
  pp. 567--583, May 2004.

\bibitem{5984917}
S.~Kaul, M.~Gruteser, V.~Rai, and J.~Kenney, ``Minimizing age of information in
  vehicular networks,'' in \emph{IEEE SECON}, 2011.

\bibitem{6195689}
S.~Kaul, R.~Yates, and M.~Gruteser, ``Real-time status: How often should one
  update?'' in \emph{IEEE INFOCOM}, 2012.

\bibitem{7282742}
L.~Huang and E.~Modiano, ``Optimizing age-of-information in a multi-class
  queueing system,'' in \emph{IEEE ISIT}, 2015.

\bibitem{6620189}
C.~Kam, S.~Kompella, and A.~Ephremides, ``Age of information under random
  updates,'' in \emph{IEEE ISIT}, 2013.

\bibitem{bestpaper}
I.~Kadota, A.~Sinha, and E.~Modiano, ``Optimizing age of information in
  wireless networks with throughput constraints,'' in \emph{IEEE INFOCOM},
  2018.

\bibitem{8006590}
Y.~P. Hsu, E.~Modiano, and L.~Duan, ``Age of information: Design and analysis
  of optimal scheduling algorithms,'' in \emph{IEEE ISIT}, 2017.

\bibitem{7852321}
I.~Kadota, E.~Uysal-Biyikoglu, R.~Singh, and E.~Modiano, ``Minimizing the age
  of information in broadcast wireless networks,'' in \emph{Allerton}, 2016.

\bibitem{DBLP:journals/corr/KaulY17}
S.~K. Kaul and R.~D. Yates, ``Status updates over unreliable multiaccess
  channels,'' \emph{CoRR}, vol. abs/1705.02521, 2017.

\bibitem{DBLP:journals/corr/abs-1803-06469}
R.~Talak, S.~Karaman, and E.~Modiano, ``Distributed scheduling algorithms for
  optimizing information freshness in wireless networks,'' \emph{CoRR}, vol.
  abs/1803.06469, 2018.

\bibitem{IT}
Q.~He, D.~Yuan, and A.~Ephremides, ``Optimal link scheduling for age
  minimization in wireless systems,'' \emph{IEEE Transactions on Information
  Theory}, vol.~64, no.~7, pp. 5381--5394, 2018.

\bibitem{schulz_et_al:LIPIcs:2016:6415}
A.~S. Schulz and J.~Verschae, ``{Min-Sum Scheduling Under Precedence
  Constraints},'' in \emph{ESA}, 2016.

\bibitem{CARRASCO2013436}
R.~A. Carrasco, G.~Iyengar, and C.~Stein, ``Single machine scheduling with
  job-dependent convex cost and arbitrary precedence constraints,''
  \emph{Operations Research Letters}, vol.~41, no.~5, pp. 436 -- 441, 2013.

\bibitem{Sidney}
J.~B. Sidney, ``Decomposition algorithms for single-machine sequencing with
  precedence relations and deferral costs,'' \emph{Operations Research},
  vol.~23, no.~2, pp. 283--298, 1975.

\bibitem{LAWLER197875}
E.~Lawler, ``Sequencing jobs to minimize total weighted completion time subject
  to precedence constraints,'' in \emph{Algorithmic Aspects of Combinatorics},
  ser. Annals of Discrete Mathematics, B.~Alspach, P.~Hell, and D.~Miller,
  Eds.\hskip 1em plus 0.5em minus 0.4em\relax Elsevier, 1978, vol.~2, pp. 75 --
  90.

\bibitem{Hu}
D.~Adolphson and T.~C. Hu, ``Optimal linear ordering,'' \emph{SIAM Journal on
  Applied Mathematics}, vol.~25, no.~3, pp. 403--423, 1973.

\bibitem{Garey:1990:CIG:574848}
M.~R. Garey and D.~S. Johnson, \emph{Computers and Intractability; A Guide to
  the Theory of NP-Completeness}.\hskip 1em plus 0.5em minus 0.4em\relax New
  York, NY, USA: W. H. Freeman \& Co., 1990.

\bibitem{smith}
W.~E. Smith, ``Various optimizers for single-stage production,'' \emph{Naval
  Research Logistics (NRL)}, vol.~3, no. 1-2, pp. 59--66, 1956.

\end{thebibliography}

\appendix
\subsection{Proof of Proposition~\ref{prop: wc}}
Consider an instance $I$ of the Min-WCS problem that has $n-1$ single-job job chains.
Each job in these single-job job chains has weight one.
The last job chain in $I$ has $L$ jobs, and the value of $L$ will be determined later.
Each of these $L$ jobs has weight two.
Hence, these $L$ jobs are completed first in $S^*_{wc}$.
Thus, $cs(S^*_{wc}) = L^2+(L+1)^2+\cdots+(L+n-1)^2$.
Consider another feasible schedule $S'$ in which these $L$ jobs are completed lastly.
Thus, $cs(S') = 1^2+2^2+\cdots+(n-1)^2+(L+n-1)^2$ and 
$wc(S') = 1+2+\cdots+(n-1)+2[n+(n+1)+\cdots+(n+L-1)] \leq 2(L+n-1)^2$.
We set $L$ to a positive integer such that $1^2+2^2+\cdots+(n-1)^2 \leq (L+n-1)^2$.
Thus, $cs(S') \leq 2(L+n-1)^2$.
On the other hand, we have
\begin{align*}
cs(S^*_{wc}) &= \frac{(L+n-1)(L+n)(2L+2n-1)}{6}-\frac{(L-1)(L)(2L-1)}{6}\\
             &\geq \frac{(L+n-1)(L+n)(2L+2n-1)}{6}-\frac{(L+n-1)(L+n)(2L+n-1)}{6}\\
             &= \frac{n(L+n-1)(L+n)}{6} \geq \frac{n(L+n-1)^2}{6}.
\end{align*}
Thus, $\displaystyle \frac{wcs(S^*_{wc})}{wcs(S')} \geq 
\frac{cs(S^*_{wc})}{cs(S')+wc(S')} \geq \frac{\frac{n}{6}(L+n-1)^2}{4(L+n-1)^2} 
= \frac{n}{24}$.

\subsection{Proof of Proposition~\ref{prop: cs}}
Consider an instance $I$ of the Min-WCS problem that has $n-1$ single-job job chains.
Each job in these single-job job chains has weight one.
The last job chain $\mathcal{C}_{l}$ in $I$ has two jobs, where the last job has weight one and 
the first job $J_{h}$ has an extremely heavy weight $w_{h}$
such that, for any two feasible schedules $S_1$ and $S_2$ of $I$,
we have $\frac{cs(S_1)}{wc(S_1)}\approx 0$, $\frac{cs(S_2)}{wc(S_2)}\approx 0$, 
and $\frac{wc(S_1)}{wc(S_2)}
\approx \frac{w_{h} \cdot S_1(J_{h})}{w_{h} \cdot S_2(J_{h})} 
= \frac{S_1(J_{h})}{S_2(J_{h})}$.
Thus, for any two feasible schedules $S_1$ and $S_2$ of $I$,
$\frac{wc(S_1)+cs(S_1)}{wc(S_2)+cs(S_2)} = 
\frac{1+\frac{cs(S_1)}{wc(S_1)}}{1+\frac{cs(S_2)}{wc(S_2)}}\cdot \frac{wc(S_1)}{wc(S_2)}
\approx \frac{wc(S_1)}{wc(S_2)} \approx \frac{S_1(J_{h})}{S_2(J_{h})}$.
In $S^*_{cs}$, $\mathcal{C}_{l}$ is completed lastly 
and thus $S^*_{cs}(J_{h}) = n$.
Consider a feasible schedule $S'$ in which $\mathcal{C}_{l}$ is completed first 
and thus $S'(J_{h}) = 1$. Hence, $\frac{wcs(S^*_{cs})}{wcs(S')} \approx n$.

\subsection{Proof of Lemma~\ref{lemma: as-job only-}}
Let $I^{as}_{job}$ be an instance obtained by removing all the dummy 
jobs from $I_{job}$. 
Hence, $I^{as}_{job}$ consists of $a$-jobs and separating jobs.
It suffices to prove the following lemma, which will be 
used again in the proof of Lemma~\ref{lemma: separating}.

\begin{lemma}\label{lemma: as-job only}
Let $S$ be any valid schedule of $I^{as}_{job}$ such that 
$S(J_1) < S(J_2) < \cdots < S(J_{m-1})$.
Then, $wc(S) = \sum_{i = 1}^{3m}{(w_i^a\sum_{j = 1}^{i}{p_j^a})} 
+ \sum_{i = 1}^{m-1}{r((m-i)B-\Delta^S_i)} 
+ \sum_{i = 1}^{m-1}{w_i(i(B+1)+\Delta^S_i)}$. 
\end{lemma}

Let $I^a_{job}$ be an instance obtained by removing 
$\mathcal{C}^*_1, \mathcal{C}^*_2, \cdots, \mathcal{C}^*_{m-1}$ and 
all the dummy jobs from $I_{job}$. Hence, $I^{a}_{job}$ consists of $a$-jobs.
Before we prove Lemma~\ref{lemma: as-job only},
We first prove the following lemma.

\begin{lemma}\label{lemma: a-job only}
Every valid schedule $S$ of $I^a_{job}$ has the 
same total weighted completion time, 
$wc(S) = \sum_{i = 1}^{3m}{(w_i^a\sum_{j = 1}^{i}{p_j^a})}$. 
\end{lemma}

\begin{proof}
All the jobs in $I^a_{job}$ are $a$-jobs, and all of them have the same 
ratio of weight to processing time.
By Smith's rule~\cite{smith}, all valid schedules of $I^a_{job}$
have the same total weighted completion time.
It is easy to see that $\sum_{i = 1}^{3m}{(w_i^a\sum_{j = 1}^{i}{p_j^a})}$ is 
the total weighted completion time of the schedule where we process $J_i^a$ 
in increasing order of $i$.
\end{proof}

We are now ready to prove Lemma~\ref{lemma: as-job only}.
Every valid schedule $S$ of $I^{as}_{job}$ 
can be obtained by inserting separating jobs
to some valid schedule $S^a$ of $I^a_{job}$.
Consider a schedule $S_1$ obtained by inserting $J_1$ to $S^{a}$ at time 
$(B+1)+\Delta^{S}_1$.
Specifically, $S_1(J_1) = (B+1)+\Delta^{S}_1$, 
and, for every $a$-job $J^a_i$ that is completed after time $B+\Delta^S_1$ 
in $S^{a}$, 
$S_1(J^a_i) = S^a(J^a_i)+1$.
Let $P'$ be the set of jobs that are completed after time $B+\Delta^S_1$ 
in $S^{a}$. Let $w(P')$ be the total weight of jobs in $P'$. 
Hence, $wc(S_1) - wc(S^a)= 
(\sum_{J^a_i \in P'}{w^a_i \cdot 1}) +w_1(B+1+\Delta^S_1)
= w(P')+w_1(B+1+\Delta^S_1)$. 
Since the total processing time of jobs in $P'$ is $mB-(S_1(J_1)-1) 
= (m-1)B-\Delta^S_1$,
the total weight of jobs in $P'$, $w(P')$, 
is equal to $r((m-1)B-\Delta^S_1)$. 
Thus, $wc(S_1) - wc(S^a)= r((m-1)B-\Delta^S_1)+w_1(B+1+\Delta^S_1)$.

We can then insert $J_2$ to $S_1$ at time $2(B+1)+\Delta^S_2$, 
where $2(B+1)+\Delta^S_2 > (B+1)+\Delta^S_1$.
Let $S_2$ be the new schedule. Thus, $S_2(J_2) = 2(B+1)+\Delta^S_2$.
Since $S_2(J_1) < S_2(J_2)$, the total processing time of jobs 
completed after $S_2(J_2)$ in $S_2$ is 
$mB-(S_2(J_2) - 2) = (m-2)B-\Delta^S_2$.
Hence, 
$wc(S_2) - wc(S_1) = r((m-2)B-\Delta^S_2)+ w_2 \cdot (2(B+1)+\Delta^S_2)$. 
Repeat the above process of inserting separating jobs to obtain $S$. 
By Lemma~\ref{lemma: a-job only}, we then get
$wc(S) = \sum_{i = 1}^{3m}{(w_i^a\sum_{j = 1}^{i}{p_j^a})} 
+ \sum_{i = 1}^{m-1}{r((m-i)B-\Delta^S_i)} 
+ \sum_{i = 1}^{m-1}{w_i(i(B+1)+\Delta^S_i)}$.

\subsection{Proof of Lemma~\ref{lemma: dummy}}
Let $J^d$ be a dummy job and let $J$ be a non-dummy job with processing time 
$p$ such that $S(J) = S(J^d)+p$.
By the assumption of $S$, such $J^d$ and $J$ must exist.
Let $S'$ be the schedule obtained by swapping the order of 
$J^d$ and $J$ in $S$. 
Hence, $S'(J) = S(J)-1$ and $S'(J^d) = S(J)$.
Note that $S'$ satisfies precedence constraints because $J$ and $J^d$ must be 
from different job chains; otherwise, $S$ is infeasible.
It suffices to show that $wcs(S') < wcs(S)$, 
because we can repeat the swapping process until
all dummy jobs are completed lastly, which can be done in polynomial time.
We first consider the case where $J$ is an $a$-job.
Let $w$ be the weight of $J$.
We have
\begin{align*}
&wcs(S') - wcs(S)\\ 
&= [w \cdot S'(J) + 1 \cdot S'(J^d) + (S'(J^d))^2] 
- [w \cdot S(J) + 1 \cdot S(J^d) + (S(J^d))^2]\\
&= [w \cdot (S(J)-1) + 1 \cdot S(J) + (S(J))^2] 
- [w \cdot S(J) + 1 \cdot (S(J)-p) + (S(J)-p)^2]\\
&= w\cdot(-1) + 1\cdot(p) + (2S(J)-p)(p)\\
&= (2S(J)-p+1)p - w < 2S(J)p-w\\
&< 2[m(B+1)+3m-1]B - 20mB(B+1)\\
&<2[m(B+1)+3m(B+1)]B - 20mB(B+1)\\ 
&= 8mB(B+1) - 20mB(B+1) < 0,
\end{align*}
where the first equality follows from the fact that jobs other than 
$J$ or $J^d$ have the same completion time in $S$ and $S'$, 
and the second inequality follows from 
the fact that $S(J) \leq m(B+1)+3m-1$,
$p < B$, and $w = pr \geq 2r = 20mB(B+1)$. 

When $J$ is a separating job, $p = 1$. We then have
\begin{align*}
&wcs(S') - wcs(S)\\ 
&= [w \cdot S'(J) + 1 \cdot S'(J^d) + (S'(J^d))^2 + (S'(J))^2]
- [w \cdot S(J) + 1 \cdot S(J^d) + (S(J^d))^2 + (S(J))^2]\\
&= [w \cdot (S(J)-1) + 1 \cdot S(J) + (S(J))^2 + (S(J)-1)^2] 
- [w \cdot S(J) + 1 \cdot (S(J)-1) + (S(J)-1)^2 + S(J)^2]\\
&= w\cdot(-1) + 1\cdot 1 = 1 - w\\ 
&\leq 1-(r-2(m-1)(B+1)+1)\\ 
&= 2(m-1)(B+1)-r < 0.
\end{align*}

\subsection{Proof of Lemma~\ref{lemma: separating}}
We first assume that $S(J_1) < S(J_2) < \cdots < S(J_{m-1})$.
By Lemma~\ref{lemma: as-job only}, we have 
\begin{align*}
&wcs(S) - Q\\ 
&=\sum_{i = 1}^{m-1}{\{[r((m-i)B-\Delta^S_i)]+[w_i(i(B+1)+\Delta^S_i)]+[(i(B+1)+\Delta^S_i)^2]\}}\\
&-\sum_{i = 1}^{m-1}{\{[r((m-i)B)]+[w_i(i(B+1))]+[(i(B+1))^2]\}}\\
&=\sum_{i=1}^{m-1}{\{[-r\Delta^S_i]
+[w_i\Delta^S_i]+[(2i(B+1)+\Delta^S_i)(\Delta^S_i)]\}}\\
&=\sum_{i=1}^{m-1}{\{(\Delta^S_i)^2+(2i(B+1)+w_i-r)\Delta^S_i\}}\\
&=\sum_{i=1}^{m-1}{\{(\Delta^S_i)^2+(2i(B+1)+(-2i(B+1)+1))\Delta^S_i\}}\\
&=\sum_{i=1}^{m-1}{\{(\Delta^S_i)^2+\Delta^S_i\}}.
\end{align*}
Note that $(\Delta^S_i)^2+\Delta^S_i \leq 0$ if and only if 
$-1 \leq \Delta^S_i \leq 0$.
Obviously, $\Delta^S_i$ must be an integer.
Next, we argue that $\Delta^S_i$ cannot be $-1$, for all $1 \leq i \leq m-1$.
First, there are $i$ separating jobs completed by time
$S(J_i)$ under $S$, for all $1 \leq i \leq m-1$.
Moreover, $S$ does not have idle time slots, 
and the processing time of every separating job is one.
Therefore, if $S(J_i) = i(B+1)-1$ 
(i.e., $\Delta^S_i = -1$),
then the total processing time of $a$-jobs that are completed before $J_i$
is $S(J_i) - i= iB-1$. Since $B$ is even, $iB-1$ is odd. 
This is impossible, because, by the assumption that
every $a_i \in L$ is even and that the processing time of $J^a_i$ is $a_i$,
the total processing time of any set of $a$-jobs must be an even number.
Hence, $\Delta^S_i \neq -1$ and 
$\min_{\Delta^S_i \in \mathbb{Z} \setminus \{-1\}}{\{(\Delta^S_i)^2+\Delta^S_i\}} = 0$, 
where the minimum happens only when $\Delta^S_i = 0$. 
Because $S$ is a feasible schedule, 
$wcs(S) - Q = \sum_{i=1}^{m-1}{\{(\Delta^S_i)^2+\Delta^S_i\}} \leq 0$.
Thus, $(\Delta^S_i)^2+\Delta^S_i = 0$ for all $1 \leq i \leq m-1$.
As a result, $\Delta^S_i = 0$ for all $1 \leq i \leq m-1$.

Notice that, we have proved that 
if $S(J_1) < S(J_2) < \cdots < S(J_{m-1})$, then $wcs(S) = Q$.
Observe that $w_1 > w_2 > \cdots > w_{m-1}$.
Hence, if $S(J_1) < S(J_2) < \cdots < S(J_{m-1})$ 
and we swap the order of any two separating jobs in $S$, 
then the objective value will increase and thus will exceed $Q$.
Thus, if $S$ is a feasible schedule of $I_{job}$ such that
all dummy jobs are completed lastly in $S$, 
then $S(J_1) < S(J_2) < \cdots < S(J_{m-1})$ must hold.

\subsection{Proof of Lemma~\ref{lemma: conti}}
Assume that there is a job chain $\overbar{\mathcal{C}}_i$ such that 
$\overbar{S}(\overbar{J}_i^{p_i^1}) 
> \overbar{S}(\overbar{J}_i^1) + (p_i^1 - 1)$. Otherwise, we are done.
Let $\overbar{J}_i^j$ be the last job in $\overbar{\mathcal{C}}_i$ 
such that $\overbar{J}_i^1, \overbar{J}_i^2, \cdots, \overbar{J}_i^j$ 
are processed contiguously under $\overbar{S}$.
Let $\overbar{J}$ be the job completed immediately after $\overbar{J}_i^j$ 
under $\overbar{S}$. Hence, $\overbar{J}$ is not in $\overbar{\mathcal{C}}_i$
and $j < p_i^1$.
Let $\overbar{S^*}$ be a new schedule obtained by swapping the order of
$\overbar{J}$ and $\overbar{J}_i^1, \cdots, \overbar{J}_i^j$ in $\overbar{S}$.
Specifically, $\overbar{S^*}(\overbar{J}) = \overbar{S}(\overbar{J}_i^1)$, and
$\overbar{S^*}(\overbar{J}_i^h) = \overbar{S}(\overbar{J}_i^h)+1$  
for all $1 \leq h \leq j$.
Jobs other than $\overbar{J}$ and $\overbar{J}_i^1, \cdots, \overbar{J}_i^j$
have the same completion time under $\overbar{S}$ and $\overbar{S^*}$.
Since $\overbar{J}$ is not from $\overbar{\mathcal{C}}_i$, 
$\overbar{S^*}$ satisfies precedence constraints.
Notice that the weights of $\overbar{J}_i^1, \cdots, \overbar{J}_i^j$ are two, 
and thus these jobs have the smallest weight in the problem instance.
Therefore, the total weighted completion time can only decrease after swapping. 
Moreover, $\overbar{J}_i^1, \cdots, \overbar{J}_i^j$ are not leaf jobs. 
Thus, the total completion time squared of all leaf jobs can only decrease 
after swapping.
Hence, $\overbar{wcs}(\overbar{S^*}) \leq \overbar{wcs}(\overbar{S})$.
Repeat the swapping process until  
$\overbar{J}_i^1, \overbar{J}_i^2, \cdots, \overbar{J}_i^{p_i^1}$ 
are scheduled contiguously.
We then apply the same procedure to other job chains that violate the 
contiguous property. Note that, once a job chain satisfies the 
contiguous property, it still satisfies the contiguous property after we 
apply the procedure to other job chains. 
Thus, the desired schedule $\overbar{S'}$ can be found in polynomial time.
\end{document}